\documentclass[12pt,a4paper]{article}

\usepackage[latin1]{inputenc}
\usepackage{amsmath,amssymb,srcltx}
\usepackage[dvips]{color}
\usepackage[margin]{fixme}
\usepackage{amsmath,amssymb,amsthm} 
\usepackage{enumerate}
\usepackage{stmaryrd}
\usepackage[colorinlistoftodos]{todonotes}
\usepackage{bbm}
\usepackage{dsfont}
\numberwithin{equation}{section}
\usepackage{fullpage}

\newcommand{\bt}{\begin{thr}{\bf Theorem. }} 
\newcommand{\satz}{\begin{thr}{\bf Theorem. }\rm} 
\newcommand\E{\mathbb{E}}
\newcommand\R{\mathbb{R}}

\newcommand\p{\mathbb{P}}
\newcommand\N{\mathbb{N}}

\newcommand\F{\mathcal{F}}

\newcommand\ofp{\Omega,\mathcal{F},\mathbb{P}}
\renewcommand\i{\infty}

\newcommand{\vr}{\varrho}
\newcommand\la{\lambda}
\newcommand\ve{\varepsilon}
\newcommand{\vp}{\varphi}
\newcommand{\cA}{\mathcal{A}}
\newcommand{\tvp}{\widetilde\varphi}
\newcommand{\hvp}{\widehat\varphi}

\newcommand{\tpsi}{\widetilde\psi}

\newcommand{\sint}{\stackrel{\mbox{\tiny$\bullet$}}{}}

\newcommand{\Var}{\operatorname{Var}} 
\DeclareMathOperator{\conv}{conv}

\newtheorem{theorem}{Theorem}[section]

\newtheorem{definition}[theorem]{Definition}
\newtheorem{lemma}[theorem]{Lemma}
\newtheorem{proposition}[theorem]{Proposition}

\theoremstyle{definition}

\newtheorem{remark}[theorem]{Remark}

\begin{document}

\title{Shadow prices for continuous processes\footnote{We would like to thank two anonymous referees for a carefully reading the paper and their remarks.}}
\author{Christoph Czichowsky\footnote{Department of Mathematics, London School of Economics and Political Science, Columbia House, Houghton Street, London WC2A 2AE, UK, {\tt c.czichowsky@lse.ac.uk}. 
        Financial support by the Swiss National Science Foundation (SNF) under grant PBEZP2\_137313 and by the European Research Council (ERC) under grant FA506041 is gratefully acknowledged.}
\and Walter Schachermayer\footnote{Fakult\"at f\"ur Mathematik, Universit\"at Wien, Oskar-Morgenstern-Platz 1, A-1090 Wien, Austria, {\tt walter.schachermayer@univie.ac.at}. 
        Support by the Austrian Science Fund (FWF) under grant P25815 and Doktoratskolleg W1245, and the European Research Council (ERC) under grant FA506041 is gratefully acknowledged.}
\and Junjian Yang\footnote{Fakult\"at f\"ur Mathematik, Universit\"at Wien, Oskar-Morgenstern-Platz 1, A-1090 Wien, Austria, {\tt junjian.yang@univie.ac.at}. 
        Financial support by the European Research Council (ERC) under grant FA506041 is gratefully acknowledged.}}

\date{\today}
\maketitle

\begin{abstract}\noindent 
  In a financial market with a continuous price process and proportional transaction costs we investigate the problem of utility maximization of terminal wealth. 
  We give sufficient conditions for the existence of a shadow price process, i.e., a least favorable frictionless market leading to the same optimal strategy and utility 
    as in the original market under transaction costs. 
  The crucial ingredients are the continuity of the price process and the hypothesis of ``no unbounded profit with bounded risk". 
  A counterexample reveals that these hypotheses cannot be relaxed. 
\end{abstract}

\vspace{-0.2cm}\newline
\vspace{-0.2cm}\newline

\section{Introduction}

In this paper, we analyze {\it continuous} $(0,\infty)$-valued stock price processes $S=(S_t)_{0\leq t\leq T}$ under proportional transaction costs $0<\lambda< 1$. 
We investigate the duality theory for portfolio optimization, sometimes also called the ``martingale method", under proportional transaction costs $\lambda$
 as initiated in the seminal paper \cite{CK96} by Cvitani\'{c} and Karatzas. 

We build on our previous paper \cite{CS14}, where the duality theory was analyzed in full generality, i.e., in the framework of c\`adl\`ag (right-continuous with left limits) processes $S$.
Our present purpose is to show that the theory simplifies considerably if we restrict ourselves to {\it continuous} processes $S$. 
More importantly, we obtain sharper results than in the general c\`adl\`ag setting on the theme of the existence of a {\it shadow price}. 
This a price process $\widehat{S}$ such that frictionless trading for this price process leads to the same optimal strategy as trading in the original market under transaction costs. 
It is folklore going back to the work of Cvitani\'{c} and Karatzas \cite{CK96} that, if the minimiser of a suitable dual problem is induced by a local martingale, rather than a supermartingale, 
  there exists a shadow price process $\widehat{S}$ in the sense of Definition \ref{shadow} below. 
Let us quote Cvitani\'{c} and Karatzas \cite{CK96} on the hypothesis that the dual optimizer is induced by a local martingale:``This assumption is a big one!''
To the best of our knowledge, previously to the present paper there have been no theorems providing sufficient conditions for this local martingale property to hold true. 
Our first main result (Theorem \ref{thsorry} below) states that, assuming that $S$ is continuous and satisfies the condition of ``no unbounded profit with bounded risk" $(NUPBR)$, 
  we may conclude --- assuming only natural regularity conditions --- that the local martingale property holds true, 
  and therefore there is a shadow price process $\widehat{S}$ in the sense of Definition \ref{shadow} below. 
For this theorem to hold true, the assumption of $(NUPBR)$ is crucial. 
It is not possible to replace it by the assumption of the existence of a consistent price system for each level $0<\mu<1$ of transaction costs (abbreviated $(CPS^{\mu})$), 
 which at first glance might seem to be the natural hypothesis in the context of transaction costs. 
The example constructed in Proposition \ref{proC1}, which constitutes the second main result of this paper, yields a continuous process $S$, satisfying $(CPS^{\mu})$ for each $0<\mu<1$,     
 and such that there is no shadow price $\widehat{S}$ in the sense of Definition \ref{shadow} below.
In fact, $S$ satisfies the stickiness condition introduced by Guasoni \cite{G06}.
 
 The paper is organized as follows. 
 In Section 2, we fix notations and formulate the problem. 
 This section mainly consists of applying the general results obtained in \cite{CS14} to the general case of c\`adl\`ag processes to the present case of continuous processes. 
 Section 3 contains the main result, Theorem \ref{thsorry}, which gives sufficient conditions for the existence of a shadow price process. 
 In Section 4, we construct the above mentioned counterexample. 
 The technicalities of this example are postponed to the Appendix.

\section{Formulation of the Problem}

We fix a time horizon $T>0$ and a continuous, $(0,\infty)$-valued stock price process $S=(S_t)_{0\leq t\leq T}$, based on and adapted to a filtered probability space 
$(\Omega,\F,(\F_t)_{0\le t\le T},\p)$, satisfying the usual conditions of right continuity and saturatedness. 
We also fix proportional transaction costs $0<\lambda<1$. 
The process $((1-\lambda)S_t, S_t)_{0\leq t\leq T}$ models the bid and ask price of the stock respectively, 
 which means that the agent has to pay a higher ask price $S_t$ to buy stock shares but only receives a lower bid price $(1-\lambda)S_t$ when selling them.  

As in \cite{CS14} we define {\it trading strategies} as $\R^2$-valued, optional, finite variation processes $\varphi = (\varphi^0_t,\varphi^1_t)_{0\leq t\leq T}$, 
modeling the holdings in units of bond and stock respectively, 
such that the following {\it self-financing} condition is satisfied:
\begin{equation} \label{selffinancing}
 \int_s^td\varphi_u^0 \leq -\int_s^tS_ud\varphi_u^{1,\uparrow}+ \int_s^t(1-\lambda)S_ud\varphi_u^{1,\downarrow},
\end{equation}
for all $0\leq s<t\leq T$. 
The integrals are defined as pathwise Riemann-Stieltjes integrals, and $\varphi^{1,\uparrow}$, $\varphi^{1,\downarrow}$ denote the components of the Jordan-Hahn decomposition of $\varphi^1$. 
We recall that a process $\varphi=(\varphi_t)_{0\leq t\leq T}$ of finite variation can be decomposed into two non-decreasing processes 
$\varphi^{\uparrow}$ and $\varphi^{\downarrow}$ such that $\varphi_t = \varphi_{0-}+\varphi^{\uparrow}_t-\varphi^{\downarrow}_t.$

There is a pleasant simplification as compared to the general setting of \cite{CS14}. 
While in the case of a c\`adl\`ag process $S$ it does make a difference whether the jumps of $\varphi$ are on the left or on the right side, 
 this subtlety does not play any role for continuous processes $S$. 
Indeed, if $\varphi$ satisfies \eqref{selffinancing}, then its left-continuous version $\varphi^l$ as well as its right-continuous version $\varphi^r$ also satisfy \eqref{selffinancing}.
Therefore, we are free to impose any of these properties. 
It turns out that the convenient choice is to impose that the process $\varphi$ is {\it right-continuous}, and therefore c\`adl\`ag, 
 which is formalized in Definition \ref{strategy} below. 
Indeed, in this case $\varphi$ is a semimartingale so that the Riemann-Stieltjes integrals in \eqref{selffinancing} may also be interpreted as It\^o integrals 
and we are in the customary realm of stochastic analysis. 
But occasionally it will also be convenient to consider the left-continuous version $\varphi^l$, which has the advantage of being predictable. 
We shall indicate if we pass to the left-continuous version $\varphi^l$. 
Again by the continuity of $S$, trading strategies can be assumed to be optional. 

\begin{definition}  \label{strategy}
 Fix the level $0<\lambda<1$ of transaction costs. 
 
 For an $\R^2$-valued process $\varphi = (\varphi^0_t,\varphi^1_t)_{0\leq t\leq T}$, we define the {\normalfont liquidation value} $V_t^{liq}(\varphi)$ at time $0\leq t\leq T$ by
 \begin{equation} \label{liqvalue}
   V_t^{liq}(\varphi) := \varphi^0_t+(\varphi_t^1)^+(1-\la)S_t-(\varphi_t^1)^-S_t.
 \end{equation}

 The process $\varphi$ is called {\normalfont admissible} if 
 \begin{equation} \label{admissible}
  V_t^{liq}(\varphi)\geq 0, 
 \end{equation}
 for all $0\leq t\leq T$.
 
 For $x>0$, we denote by $\mathcal{A}(x)$ the set of admissible, $\R^2$-valued, optional, c\`adl\`ag, finite variation processes $\varphi = (\varphi^0_t,\varphi^1_t)_{0\leq t\leq T}$, 
 starting with initial endowment $(\varphi_{0-}^0,\varphi_{0-}^1)=(x,0)$ and satisfying the self-financing condition \eqref{selffinancing}.
\end{definition}

As we deal with the {\it right-continuous} processes $\varphi$, we have the usual notational problem of a jump at time zero. 
This is done by distinguishing between the value $\varphi_{0-}=(x,0)$ above and $\varphi_0=(\varphi_0^0,\varphi_0^1)$.  
In accordance with \eqref{selffinancing}, we must have 
\begin{equation*}
 \varphi_0^0-\varphi_{0-}^0 \leq -S_0(\varphi^1_0-\varphi^1_{0-})^+ + (1-\lambda)S_0(\varphi^1_0-\varphi^1_{0-})^-
\end{equation*}
i.e.,
\begin{equation*}
 \varphi_0^0 \leq x-S_0(\varphi^1_0)^+ + (1-\lambda)S_0(\varphi^1_0)^-.
\end{equation*}

We can now define the (primal) utility maximization problem. 
Let $U:\R_+\to\R$ be an increasing, strictly concave, and smooth function, satisfying the 
Inada conditions $U'(0)=\infty$ and $U'(\infty)=0$, 
as well as the condition of ``reasonable asymptotic elasticity" introduced in \cite{KS99}
\begin{equation}\label{RAE}
 \textnormal{AE}(U):= \limsup\limits_{x\to\i} \frac{xU'(x)}{U(x)} < 1.
\end{equation}

Denote by $\mathcal{C}(x)$ the convex subset in $L_+^0$ 
 \begin{equation} \label{setC}
   \mathcal{C}(x):= \left\{V^{liq}_T(\varphi):\,\,\varphi\in\mathcal{A}(x)\right\},
 \end{equation}
which equals the set $\mathcal{{C}}(x)$ as defined in \cite{CS14}.

For given initial endowment $x>0$, the agent wants to maximize {\it expected utility at terminal time $T$}, i.e.,
 \begin{equation}\label{J5}
   \E[U(g)] \to \max!, \qquad g\in\mathcal{C}(x).
 \end{equation}

In our search for a duality theory, we have to define the dual objects. 
The subsequent definition formalizes the concept of {\it consistent price processes}. 
It was the insight of Jouini and Kallal \cite{JK95} that 
this is the natural notion which, in the case of transaction costs, corresponds to the concept of equivalent martingale measures in the frictionless case. 

\begin{definition} \label{CPS}
Fix $0<\la <1$ and the continuous process $S=(S_t)_{0\le t\le T}$ as above.
A \textnormal{$\la$-consistent price system} is a two dimensional strictly positive process $Z=(Z^0_t,Z^1_t)_{0\le t\le T}$ 
 with $Z^0_0=1$, that consists of a martingale $Z^0$ and a local martingale $Z^1$ under $\p$ such that 
  \begin{equation}\label{J10}
    \widetilde{S}_t:=\frac{Z^1_t}{Z^0_t} \in [(1-\la)S_t, S_t],\qquad a.s. 
  \end{equation}
 for $0\le t\le T$.

We denote by $\mathcal{Z}^e (S)$ the set of $\la$-consistent price systems. By $\mathcal{Z}^a(S)$ we denote the set of processes $Z$ as above, which are only required to be
non-negative \textnormal{(}where we consider \eqref{J10} to be satisfied if $\frac{Z^1_t}{Z^0_t} =\frac{0}{0}$\textnormal{)}.

We say that $S$ satisfies the condition $(CPS^\la)$ of \textnormal{admitting a $\la$-consistent price system}, if $\mathcal{Z}^e(S)$ is non-empty. 

We say that $S$ satisfies \textnormal{locally} the condition $(CPS^\la)$, 
if there exists a strictly positive process $Z$ and a sequence $(\tau_n)^\i_{n=1}$ of $[0,T] \cup\{\i\}$-valued stopping times, increasing to infinity, 
such that each stopped process $Z^{\tau_n}$ defines a consistent price system for the stopped process $S^{\tau_n}$. 
\end{definition}

\begin{remark} \label{localnature}
 The central question of this paper, namely the existence of a shadow price, turns out to be of a \textnormal{local} nature.
 Hence the condition of $S$ satisfying $(CPS^{\lambda})$ \textnormal{locally} will turn out to be the natural one $($compare Definition \ref{DefDeflators} below$)$. 
 This is analogous to the frictionless setting where $(NUPBR)$, which is the local version of the condition of ``no free lunch with vanishing risk" $(NFLVR)$, 
  turns out to be the natural assumption for utility maximization problems. 
\end{remark}

The definitions above have been chosen in such a way that the following result which is analogous to \cite[Theorem 3.5]{CS06} holds true. 
For an explicit proof in the present setting see \cite[Lemma A.1]{CS14}. 

\begin{theorem}  \label{bonjovi}
Fix $x>0$, transaction costs $0<\la <1,$ and the continuous process $S=(S_t)_{0\le t\le T}$ as above. 
Suppose that $S$ satisfies $(CPS^{\mu})$ locally for all $0<\mu<\lambda$.

Then the convex set $\mathcal{C}(x)$ in $L^0_+(\ofp)$ is closed and bounded with respect to the topology of convergence in measure.

More precisely, $\mathcal{C}(x)$ has the following \textnormal{convex compactness} property \textnormal{(}compare \textnormal{\cite[Proposition 2.4]{Z10}):} 
given a sequence $(g^n)^\i_{n=1}$ in $\mathcal{C}(x)$, 
there is a sequence $(\widetilde{g}^n)^\i_{n=1}$ of forward convex combinations $\widetilde{g}^n\in\conv(g^n,g^{n+1},\cdots)$ 
such that $(\widetilde{g}^n)^\i_{n=1}$ converges a.s.~to some $g\in\mathcal{C}(x).$ \hfill $\square$
\end{theorem}

The main message is the closedness (resp.~the convex compactness) property of the set $\mathcal{C}(x)$ of {\it attainable claims} 
over which we are going to optimize.
It goes without saying that such a closedness property is of fundamental importance when we try to optimize over $\mathcal{C}(x)$ as in \eqref{J5}.
In the frictionless case, such a closedness property is traditionally obtained under the assumption of ``no free lunch with vanishing risk'' (compare \cite{DS94}). 
It was notably observed by Karatzas and Kardaras \cite{KK07} (in the frictionless setting) that -- as mentioned in Remark \ref{localnature} --
 it is sufficient to impose this property only {\it locally} when we deal with trading strategies which at all times have a non-negative value. 
Compare also \cite{CS96}, \cite{SY98}, \cite{GK03} and \cite{TS14}.

Similarly, in the present setting of Theorem \ref{bonjovi} it turns out that it suffices to impose a local assumption, 
 namely the {\it local} assumption of $(CPS^{\mu})$, for all $0<\mu <\la$, 
 as has been observed by \cite{BY13}.

\vspace{3mm}

We now translate Definition \ref{CPS} into the language of local and supermartingale deflators as introduced in \cite{KK07} and \cite{K12} in the frictionless setting, and 
  in \cite{BY13} and \cite{CS14} in the setting of transaction costs.  

\begin{definition} \label{DefDeflators}
Fix $0<\la <1$ and the continuous process $S=(S_t)_{0\le t\le T}$ as above. 

The set $\mathcal{Z}^{loc,e}$ (resp. $\mathcal{Z}^{loc,a}$) of \textnormal{$\la$-consistent local martingale deflators} consists of the strictly positive (resp.~non-negative) 
processes $Z=(Z^0_t,Z^1_t)_{0\leq t\leq T}$, normalized by $Z^0_0=1$, such that there exists a localizing sequence $(\tau_n)^\i_{n=1}$ of stopping times so that $Z^{\tau_n}$ is in $\mathcal{Z}^e$ 
(resp. $\mathcal{Z}^a$) for the stopped process $S^{\tau_n}$.

The set $\mathcal{Z}^{sup,e}$ (resp. $\mathcal{Z}^{sup,a}$) of \textnormal{$\la$-consistent supermartingale deflators} consists of the strictly positive
(resp.~non-negative) processes $Y=(Y^0_t,Y^1_t)_{0\leq t\leq T}$, normalized by $Y^0_0=1$, such that $\frac{Y^1}{Y^0}$ takes values in $[(1-\lambda)S,S]$ and
such that, for every $\varphi=(\varphi^0_t,\varphi^1_t)_{0\le t\le T}\in\mathcal{A}(1)$, 
the value process
\begin{equation}\label{J12}
  V_t:=\varphi^0_t Y^0_t+\varphi^1_t Y^1_t
\end{equation}
is a supermartingale under $\p.$
\end{definition}
 
 Contrary to \cite{CS14}, where we were forced to consider optional strong supermartingales, in the present setting of continuous $S$ 
  we may remain in the usual realm of (c\`adl\`ag) supermartingales in the above definition - compare Proposition \ref{proJ37} and Proposition \ref{proJ38}.  
 We use the letter $Y$ to denote supermartingales rather then the letter $Z$, which will be reserved to (local) martingales. 

 Obviously $\mathcal{Z}^{loc,e}\neq\emptyset$, for $0<\lambda<1$, amounts to requiring that $(CPS^{\lambda})$ holds true locally. 

Using the notation $\widetilde{S}_t=\frac{Y^1_t}{Y^0_t}$ as in \eqref{J10}, we may rewrite the value process $V_t$ as
\begin{equation}\label{J12a}
  V_t=Y^0_t \big(\varphi^0_t+\varphi^1_t\widetilde{S}_t\big).
\end{equation}

Comparing \eqref{J12a} to the liquidation value $V^{liq}_t$ in \eqref{liqvalue}, we infer that $V_t \geq V_t^{liq}$ as $\widetilde{S}$ takes values in $[(1-\la)S, S]$. 
The admissibility condition \eqref{admissible} therefore implies the non-negativity of $(V_t)_{0\le t\le T}$. 
Looking at formula \eqref{J12a} one may interpret $\widetilde{S}_t$ as a {\it valuation} of the stock position $\varphi^1_t$
by some element in the bid-ask spread $[(1-\la)S_t, S_t],$ while $Y^0_t$ plays the role of a {\it deflator}, well known from the frictionless theory.

\vspace{3mm}

The next result states the rather obvious fact that supermartingale deflators are a generalization of local martingale deflators. 
It will be proved in the Appendix.

\begin{proposition}\label{queen}
Fix $0<\la <1$ and a continuous process $S=(S_t)_{0\le t\le T}$. Then
\begin{equation*}
\mathcal{Z}^{loc,e} \subseteq \mathcal{Z}^{sup,e} \quad \mbox{and} \quad \mathcal{Z}^{loc,a} \subseteq \mathcal{Z}^{sup,a},
\end{equation*}
i.e., a $\la$-consistent local martingale deflator is a $\la$-consistent supermartingale deflator.
\end{proposition}

We now are in a position to define the set $\mathcal{D}=\mathcal{D}(1)$ of {\it dual variables} which will turn out to be {\it polar} to the set $\mathcal{C}=\mathcal{C}(1)$
of primal variables defined in \eqref{setC}.

\begin{definition}\label{defJ14}
 Fix $0<\la <1$ and a continuous process $S=(S_t)_{0\le t\le T}$ as above.
 For $y>0$ we denote by $\mathcal{B}(y)$ the set of supermartingale deflators $Y=(Y^0_t,Y^1_t)_{0\leq t\leq T}$, starting at $Y_0^0=y$. 
 More formally, $\mathcal{B}(y)=y\mathcal{Z}^{sup,a}$ consists of all non-negative supermartingales $Y=(Y^0_t,Y^1_t)_{0\leq t\leq T}$
 such that $Y^0_0=y$ and
     $$\frac{Y^1_t}{Y^0_t}\in[(1-\lambda)S_t,S_t]$$  for all $0\leq t\leq T$, 
 and such that $\varphi^0Y^0+\varphi^1Y^1$ is a supermartingale for all $(\varphi^0,\varphi^1)\in\mathcal{A}(1)$. 
 We denote by $\mathcal{B}$ the set $\mathcal{B}(1)$.

 We denote by $\mathcal{D}(y)$ the set of random variables $h\in L^0_+(\Omega,\F,\p)$ such that there is a supermartingale deflator 
 $(Y^0_t,Y^1_t)_{0\le t\le T}\in\mathcal{B}(y)$,
 whose first coordinate has terminal value $Y_T^0=h$. We denote by $\mathcal{D}$ the set $\mathcal{D}(1)$.
\end{definition}

The definition of supermartingale deflators is designed so that the following closedness property holds true.
A subset $A$ in $L_+^0$ is called solid, if $g\in A$ and $0\leq h\leq g$ implies that $h\in A$. 

\begin{proposition} \label{dan}
Fix $0<\la <1$ and the continuous process $S=(S_t)_{0\le t\le T}$. 

Then the set $\mathcal{D}$ is a convex, solid subset of $L^1_+(\ofp),$ bounded in norm by one, and \textnormal{closed with respect to convergence in measure}. 
In fact, for a sequence $(h^n)^\i_{n=1}$ in $\mathcal{D}$, there is a sequence of convex combinations $\widetilde{h}^n \in \conv (h^n,h^{n+1},\cdots)$ 
such that $(\widetilde{h}^n)^\i_{n=1}$ converges a.s.~to some $h\in \mathcal{D}.$ \hfill $\square$
\end{proposition}

This proposition goes back to \cite{KS99} and was explicitly stated and proved in the frictionless case in \cite{KK07}. 
In the present transaction cost setting, it was proved in \cite[Lemma A.1]{CS14} in the framework of c\`adl\`ag processes. 

\vspace{3mm}

Now we can state the polar relation between $\mathcal{C}$ and $\mathcal{D}$. In fact, these sets satisfy {\it verbatim} the conditions isolated in (\cite[Proposition 3.1]{KS99}).
Here is the precise statement. For an explicit proof in the present setting see \cite[Lemma A.1]{CS14}. 

\begin{proposition}   \label{proJ16}
 Fix the continuous process $S=(S_t)_{0\le t\le T}$ and $0<\la <1$. 
 Suppose that $S$ satisfies $(CPS^{\mu})$ locally for all $0<\mu<\lambda$.
 We then have:
\begin{enumerate}[(i)]
 \item The sets $\mathcal{C}$ and $\mathcal{D}$ are solid, convex subsets of $L^0_+=L^0_+(\ofp)$ which are closed with respect to convergence in measure. \\
       Denoting by $\mathcal{D}^{loc}$ the set of terminal values $Z^0_T$, 
        where $Z$ ranges in $\mathcal{Z}^{loc,e}$, the set $\mathcal{D}$ equals the closed, solid hull of $\mathcal{D}^{loc}$.
 \item For $g\in L^0_+$ we have that $g\in\mathcal{C}$ iff we have $\E[gh] \le 1,$ for all $h\in \mathcal{D}$. \\
       For $h\in L^0_+$ we have that $h\in\mathcal{D}$ iff we have $\E[gh] \le 1,$ for all $g\in\mathcal{C}$.
 \item The set $\mathcal{C}$ is bounded in $L^0$ and contains the constant function $\mathbf{1}.$ \hfill $\square$
\end{enumerate}
\end{proposition}

We can now conclude from the above Proposition \ref{proJ16} that the theorems of the duality theory of portfolio optimization, as obtained in 
\cite[Theorem 3.1 and 3.2]{KS99}, carry over {\it verbatim} to the present setting as these theorems only need the validity of this proposition as input. 
We recall the essence of these theorems.

\begin{theorem}[Duality Theorem] \label{thhurt}
In addition to the hypotheses of Proposition \ref{proJ16} suppose that there is a utility function $U:(0,\infty)\to \R$ satisfying \eqref{RAE} above. 
Define the primal and dual value function as
\begin{align}\label{J18}
  u(x):=\sup\limits_{g\in\mathcal{C}(x)} \E[U(g)], \\
  v(y):=\inf\limits_{h\in\mathcal{D}(y)} \E[V(h)], \label{geht}
\end{align}
where  $$V(y):= \sup_{x>0}\{U(x)-xy\}, \quad y>0 $$ 
 is the conjugate function of $U$, and suppose that $u(x) <\i,$ for some $x>0.$ 
Then the following statements hold true.
\begin{enumerate}[(i)]
 \item The functions $u(x)$ and $v(y)$ are finitely valued, for all $x,y>0$, and mutually conjugate
        \begin{equation*}
             v(y)=\sup\limits_{x>0} [u(x)-xy], \quad u(x) =\inf\limits_{y>0} [v(y)+xy].
        \end{equation*}
       The functions $u$ and $v$ are continuously differentiable and strictly concave (resp.~convex) and satisfy
        \begin{equation*}
             u'(0) =-v'(0)=\i, \qquad u'(\i)=v'(\i)=0.
        \end{equation*}
        
 \item For all $x,y>0$, the solutions $\widehat{g}(x)\in\mathcal{C}(x)$ in \eqref{J18} and $\widehat{h}(y)\in\mathcal{D}(y)$ in \eqref{geht} exist, 
        are unique and take their values a.s.~in $(0,\i)$.
       There are $\big(\widehat{\varphi}^0(x),\widehat{\varphi}^1(x)\big)\in\mathcal{A}(x)$ and $\big(\widehat{Y}^0(y),\widehat{Y}^1(y)\big)\in\mathcal{B}(y)$
        such that 
        $$ V_T^{liq}\big(\widehat{\varphi}(x)\big)=\widehat{g}(x)\quad\mbox{and}\quad \widehat{Y}^0_T(y)=\widehat{h}(y). $$
 \item  If $x>0$ and $y>0$ are related by $u'(x)=y$, or equivalently $x=-v'(y)$, 
        then $\widehat{g}(x)$ and $\widehat{h}(y)$ are related by the first order conditions
        \begin{equation} \label{Y=U'(X)}
            \widehat{h}(y)=U'\big(\widehat{g}(x)\big) \quad\mbox{and}\quad \widehat{g}(x) = -V'\big(\widehat{h}(y)\big),
        \end{equation}  
        and we have that
        \begin{equation} \label{xy=EXY}
           \E\big[\widehat{g}(x)\widehat{h}(y)\big]=xy.
        \end{equation}       
        In particular, the process $\widehat{\varphi}^0_t(x)\widehat{Y}^0_t(y)+\widehat{\varphi}^1_t(x)\widehat{Y}^1_t(y)$ is a uniformly integrable $\p$-martingale.
\end{enumerate}
\end{theorem} 

After these preparations, which are variations of known results, we now turn to the central topic of this paper.

The Duality Theorem \ref{thhurt} asserts the existence of a strictly positive dual optimizer $\widehat{h}(y)\in\mathcal{D}(y)$, 
which implies that there is an {\it equivalent supermartingale deflator} $\widehat{Y}(y)=\big(\widehat{Y}^0_t(y),\widehat{Y}^1_t(y)\big)_{0\le t\le T}\in\mathcal{B}(y)$ 
such that $\widehat{h}(y)=\widehat{Y}^0_T(y)$.
We are interested in the question whether the {\it supermartingale} $\widehat{Y}(y)$ can be chosen to be a {\it local martingale}. 
We say ``can be chosen'' for the following reason: 
it follows from $(ii)$ above that the first coordinate $\widehat{Y}^0(y)$ of $\widehat{Y}(y)$ is uniquely determined; 
but we made no assertion on the uniqueness of the second coordinate $\widehat{Y}^1(y)$.

The phenomenon that the dual optimizer may be induced by a supermartingale only, rather than by a local martingale, 
is well-known in the frictionless theory (\cite[Example 5.1 and 5.1']{KS99}). 
This phenomenon is related to the singularity of the utility function $U$ at the left boundary of its domain,
where we have $U'(0):=\lim_{x\searrow 0} U'(x)=\i$. 
If one passes to utility functions $U$ which take finite values on the entire real line, e.g., $U(x)=-e^{-x}$, 
the present ``{\it supermartingale phenomenon}'' does not occur any more (compare \cite{WS01}).

In the present context of portfolio optimization under transaction costs, the question of the local martingale property of the dual optimizer $\widehat{Y}(y)$ is
of crucial relevance in view of the subsequent {\it Shadow Price Theorem}. 
It states that, if the dual optimizer is induced by a local martingale, there is a shadow price. 
This theorem essentially goes back to the work of Cvitani\'{c} and Karatzas \cite{CK96}. 
While these authors did not explicitly crystallize the notion of a shadow price, 
 subsequently Loewenstein \cite{L00} explicitly formulated the relation between a financial market under transaction costs and a corresponding frictionless market.  
Later this has been termed ``shadow price process'' (compare also \cite{KMK11,BCKMK13} as well as \cite{KMK10,GMKS13,GGMKS14,CSZ13,HP11} for constructions in the Black--Scholes model).

We start by giving a precise meaning to this notion (see also \cite[Definition 2.1.]{CS14}).

\begin{definition} \label{ShadowPriceDef}
 In the above setting a semimartingale $\widetilde{S}=(\widetilde{S}_t)_{0\le t\le T}$ is called a \textnormal{shadow price process} 
  for the optimization problem \eqref{J5} if
 \begin{enumerate}[(i)]
  \item $\widetilde{S}$ takes its values in the bid-ask spread $[(1-\la)S, S]$. 
  \item The optimizer to the corresponding frictionless utility maximization problem 
      \begin{equation} \label{frictionlessP}
         \E[U(\widetilde{g})]\to\max!, \qquad \widetilde{g}\in\widetilde{\mathcal{C}}(x), 
      \end{equation}
      exists and coincides with the solution $\widehat{g}(x)\in \mathcal{C}(x)$ for the optimization problem \eqref{J5} under transaction costs. 
      In \eqref{frictionlessP} the set $\widetilde{\mathcal{C}}(x)$ consists of all non-negative random variables, which are attainable by starting with initial endowment $x$ and 
      then trading the stock price process $\widetilde{S}$ in a frictionless admissible way, as defined in \textnormal{\cite{KS99}}.  
  \item The optimal trading strategy $\widehat{H}$ \textnormal{(}in the sense of predictable, $\widetilde{S}$-integrable process for the frictionless market $\widetilde{S}$, 
        as in \textnormal{\cite{KS99}}\textnormal{)} 
        is equal to the left-continuous version of the finite variation process $\widehat{\varphi}^1(x)$ of the unique optimizer 
        $(\widehat{\varphi}^0_t(x),\widehat{\varphi}^1_t(x))_{0\leq t\leq T}$ of the optimization problem \eqref{J5}. 
 \end{enumerate}
\end{definition}

 The essence of the above definition is that the value function $\tilde{u}(x)$ of the optimization problem for the frictionless market $\widetilde{S}$ is equal to the 
  value function $u(x)$ of the optimization problem for $S$ under transaction costs, i.e.,  
      \begin{equation}\label{J22}
        \tilde{u}(x):= \sup_{\widetilde{g} \in \widetilde{\mathcal{C}}(x)} \E\big[U\big(\widetilde{g}\big)\big] 
                     = \sup_{g\in \mathcal{C}(x)} \E [U(g)]= u(x), 
      \end{equation}
  although the set $\widetilde{\mathcal{C}}(x)$ contains the set $\mathcal{C}(x)$ defined in \eqref{setC}. 

 The subsequent theorem was proved in the framework of general c\`adl\`ag processes in \cite[Proposition 3.7]{CS14}. 

\begin{theorem}[Shadow Price Theorem] \label{shadow}
Under the hypothesis of Theorem \ref{thhurt} fix $x>0$ and $y>0$ such that $u'(x)=y$. 
Assume that the dual optimizer $\widehat{h}(y)$ equals $\widehat{Y}^0_T(y)$, 
where $\widehat{Y}(y)\in\mathcal{B}(y)$ is a \textnormal{local $\p$-martingale}.

Then the strictly positive semimartingale $\widehat{S}:=\frac{\widehat{Y}^1(y)}{\widehat{Y}^0(y)}$ is a {\it shadow price process} 
$($in the sense of Definition \ref{ShadowPriceDef}$)$ for the optimization problem \eqref{J5}. \hfill $\square$
\end{theorem}

\begin{remark}
 Let $\widehat{S}$ be a the shadow price process as above and 
 define the optional sets in $\Omega\times [0,T]$
 \begin{equation*}
   A^{buy} =\left\{\widehat{S}_t =S_t\right\} \quad\mbox{and} \quad A^{sell} =\left\{\widehat{S}_t= (1-\la)S_t\right\}. 
 \end{equation*}
 The optimizer $\widehat{\varphi}=(\widehat{\varphi}^0,\widehat{\varphi}^1)$ of the optimization problem \eqref{J5} for $S$ under transaction costs $\la$ satisfies 
  \begin{align*}
   &\left\{d\widehat{\varphi}^1_t(x)<0\right\}\subseteq \left\{\widehat{S}_t=(1-\lambda)S_t\right\},  \\
   &\left\{d\widehat{\varphi}^1_t(x)>0\right\}\subseteq \left\{\widehat{S}_t=S_t\right\} , 
 \end{align*}
  for all $0\leq t\leq T$, 
  i.e., the measures associated to the increasing process $\widehat{\varphi}^{1,\uparrow}$ (respectively $\widehat{\varphi}^{1,\downarrow}$) are supported by $A^{buy}$ (respectively $A^{sell}$). 
 This crucial feature has been originally shown by Cvitani\'c and Karatzas \textnormal{\cite{CK96}} in an It\^{o} process setting. 
 In the present form, it is a special case of \textnormal{\cite[Theorem 3.5]{CS14}}.
\end{remark}

\section{The Main Theorem}

In the Shadow Price Theorem \ref{shadow}, we simply {\it assumed} that the the dual optimizer $\widehat{Y}_T(y)$ is induced by a {\it local martingale} 
 $\widehat{Z}=(\widehat{Z}^0_t, \widehat{Z}^1_t)_{0\le t\le T}$.
In this section, we present our main theorem, which provides sufficient conditions for this local martingale property to hold true. 

For the convenience of the reader, we recall the definition of the condition $(NUPBR)$ of ``no unbounded profit with bounded risk'', which is the key condition in the main theorem. 

\begin{definition}
 A semimartingale $S$ is said to satisfy the condition $(NUPBR)$ of ``no unbounded profit with bounded risk'', if the set 
  $$ \{(H\sint S)_T:\, H \mbox{ is 1-admissible strategy}\} $$
  is bounded in $L^0$, where 
  $$(H\sint S)_t = \int_0^t H_udS_u,  \quad 0\leq t\leq T, $$ denotes the stochastic integral with respect to $S$. 
\end{definition}

\begin{theorem}\label{thsorry}
 Fix the level $0<\la <1$ of transaction costs and assume that the assumptions of Theorem \ref{thhurt} plus the assumption of $(NUPBR)$ are satisfied. 
 To resume: $S=(S_t)_{0\le t\le T}$ is a \textnormal{continuous}, strictly positive semimartingale satisfying the condition $(NUPBR)$,   
  and $U:(0,\infty)\to\R$ is a utility function satisfying the condition \eqref{RAE} of reasonable asymptotic elasticity. 
 We also suppose that the value function $u(x)$ in \eqref{J18} is finite, for some $x>0$.

Then, for each $y>0$, the dual optimizer $\widehat{h}(y)$ in Theorem \ref{thhurt} is induced by a \textnormal{local martingale} $\widehat{Z}=(\widehat{Z}^0_t,\widehat{Z}^1_t)_{0\le t\le T}$.
Hence, by Theorem \ref{shadow}, the process $\widehat{S}:=\frac{\widehat{Z}^1}{\widehat{Z}^0}$ is a shadow price. 
\end{theorem}

Before proving the theorem, let us comment on its assumptions. The {\it continuity assumption} on $S$ cannot be dropped. 
A two-period counterexample was given in \cite{BCKMK13}, and a more refined version in \cite{CMKS14}. 
These constructions are ramifications of Example 6.1' in \cite{KS99}.

The assumption of $S$ satisfying $(NUPBR)$, which is the local version of the customary assumption $(NFLVR)$, 
is quite natural in the present context. 
Nevertheless one might be tempted (as the present authors originally have been) to conjecture that this assumption could be replaced by a weaker assumption as used in Proposition \ref{proJ16}, 
i.e., that for every $0<\mu <\la$ there exists a $\mu$-consistent price system, at least locally. 
This would make the above theorem applicable also to price processes which fail to be semimartingales, e.g., processes based on fractional Brownian motion.
Unfortunately, this idea was wishful thinking and such hopes turned out to be futile. 
In Proposition \ref{proC1} below, we give a counterexample showing the limitations of Theorem \ref{thsorry}. 

\medskip

Turning to the proof of Theorem \ref{thsorry}, we split its message into the two subsequent propositions which clarify where the assumption of $(NUPBR)$ is crucially needed.
We now drop $x$ and $y$ from $\widehat{\varphi}(x)$ and $\widehat{Y}(y)$, respectively, for the sake of simplicity. 

\begin{proposition}\label{proJ37}
Fix $0<\la <1$. 
Under the assumptions of Theorem \ref{thhurt} $($where we do not impose the assumption $(NUPBR))$,
suppose that the liquidation value process associated to the optimizer 
$\widehat{\varphi}=\big(\widehat{\varphi}_t^0,\widehat{\varphi}_t^1\big)_{0\leq t\leq T}$
\begin{equation}\label{J37}
   \widehat{V}_t^{liq}: =\widehat{\varphi}^0_t +(1-\la)(\widehat{\varphi}^1_t)^+ S_t-(\widehat{\varphi}^1_t)^- S_t
\end{equation}
is strictly positive, almost surely for each $0\leq t\leq T$.

Then the assertion of Theorem \ref{thsorry} holds true, i.e., the dual optimizer $\widehat{h}(y)$ is induced by a \textnormal{local martingale}
$\widehat{Z}=(\widehat{Z}^0_t,\widehat{Z}^1_t)_{0\le t\le T}$.
\end{proposition}

\begin{proposition}\label{proJ38}
Under the assumptions of Theorem \ref{thsorry}, i.e., $S$ is a continuous semimartingale satisfying the condition $(NUPBR)$, 
the liquidation value process $\widehat{V}_t^{liq}$ in \eqref{J37} is strictly positive, i.e., $\inf_{0\leq t\leq T}\widehat{V}^{liq}_t>0$ almost surely. 
\end{proposition}

\vspace{3mm}

Obviously Proposition \ref{proJ37} and \ref{proJ38} imply Theorem \ref{thsorry}.
We start with the proof of the second proposition.

\begin{proof}[Proof of Proposition \ref{proJ38}]
 As shown by Choulli and Stricker \cite[Th\'{e}or\`{e}me 2.9]{CS96} (compare also \cite{KK07,K12,SY98,TS14}),
 the condition $(NUPBR)$ implies the existence of a {\it strict martingale density} for the continuous semimartingale $S$, 
 i.e., a $(0,\infty)$-valued local martingale $Z$ such that $ZS$ is a local martingale. 
 Note that $\big(\widehat{V}_t^{liq}\big)_{0\leq t\leq T}$ is a semimartingale as we assumed $\varphi$ to be optional and c\`adl\`ag,  
  which makes the application of It\^o's lemma legitimate.  
 Applying It\^o's lemma to the semimartingale $Z\widehat{V}^{liq}$ and recalling that $\varphi$ has finite variation, 
  we get from \eqref{J37}
\begin{align*}
  d (Z_t \widehat{V}_t^{liq}) =& Z_{t}\left[d\widehat{\varphi}^0_t +(1-\la)S_td(\widehat{\varphi}^1_t)^+ -S_t d(\widehat{\varphi}^1_t)^-\right]  \\ 
                     & +\widehat{\varphi}^0_{t-}dZ_t +\left[ (1-\la) (\widehat{\varphi}^1_{t-})^+ -(\widehat{\varphi}^1_{t-})^-\right] d(Z_t S_t). \notag
\end{align*}

By \eqref{selffinancing}, the increment in the first bracket is non-positive. 
The two terms $dZ_t$ and $d(Z_t S_t)$ are the increments of a local martingale. 
Therefore the process $Z \widehat{V}^{liq}$ is a local supermartingale under $\p$. As $Z\widehat{V}^{liq} \geq 0$, it is, in fact, a supermartingale.

Since $Z_T$ is strictly positive and the terminal value $\widehat{V}^{liq}_T$ is strictly positive a.s.~by Theorem \ref{thhurt}, 
 we have that the trajectories of $Z\widehat{V}^{liq}$ are a.s.~strictly positive, by the supermartingale property of $Z\widehat{V}^{liq}$. 
 This implies that the process $\widehat{V}^{liq}$ is a.s.~strictly positive (compare \cite[Theorem 1.7]{WS14}).
\end{proof}

\medskip

\begin{proof}[Proof of Proposition \ref{proJ37}]
 Fix $y>0$ and assume without loss of generality that $y=1$. 
 We have to show that there is a local martingale deflator $\widehat{Z}=(\widehat{Z}^0_t,\widehat{Z}^1_t)_{0\le t\le T}$ with $\widehat{Z}^0_0=1$ and $\widehat{Z}^0_T=\widehat{h}$, 
 where $\widehat{h}$ is the dual optimizer in Theorem \ref{thhurt} for $y=1$. 
 
 By Proposition \ref{proJ16} $(i)$, we know that there is a sequence $(Z^n)_{n=1}^{\infty}$ of local martingale deflators such that 
 \begin{equation*}
  \lim_{n\to\i} Z^{0,n}_T \geq \widehat{h}, \quad a.s.
 \end{equation*}
 By the optimality of $\widehat{h}$, we must have equality above. 
 Using Lemma \ref{attersee} below, we may assume, by passing to convex combinations, that the sequence $(Z^n)^\i_{n=1}$ converges to a supermartingale, 
  denoted by $\widehat{Z}$, in the sense of \eqref{Ee3a}.
  
 By passing to a localizing sequence of stopping times, we may assume that all processes $Z^n$ are uniformly integrable martingales, 
 that $S$ is bounded from above and bounded away from zero, and that the process $\widehat{\varphi}$ is bounded.
 
 To show that the supermartingale $\widehat{Z}$ is a local martingale, consider its Doob-Meyer decomposition 
 \begin{align}\label{fo}
   d\widehat{Z}^0_t = d\widehat{M}^0_t -d\widehat{A}^0_t, \\
   d\widehat{Z}^1_t = d\widehat{M}^1_t -d\widehat{A}^1_t, \label{po}
 \end{align}
 where the predictable processes $\widehat{A}^0$ and $\widehat{A}^1$ are non-decreasing. 
 We have to show that $\widehat{A}^0$ and $\widehat{A}^1$ vanish.
 By stopping once more, we may assume that these two processes are bounded and 
 that $\widehat{M}^0$ and $\widehat{M}^1$ are true martingales. 

 We start by showing that $\widehat{A}^0$ and $\widehat{A}^1$ are aligned in the following way
 \begin{equation}\label{claim}
   (1-\la) S_t d\widehat{A}^0_t \le d\widehat{A}^1_t \le S_t d\widehat{A}^0_t,
 \end{equation}
 which is the differential notation for the integral inequality
  \begin{equation}  \label{int1D}
    \int^{T}_{0} (1-\la) S_t\mathbf{1}_{D}d\widehat{A}^0_t \le \int^{T}_{0} \mathbf{1}_{D} d\widehat{A}^1_t \le \int^{T}_{0} S_t\mathbf{1}_{D} d\widehat{A}^0_t,
  \end{equation}
  which we require to hold true for every optional subset $D\subseteq\Omega\times[0,T]$.
 Turning to the differential notation again, inequality \eqref{claim} may be intuitively interpreted that $\frac{d\widehat{A}^1_t}{d\widehat{A}^0_t}$ takes values in the bid-ask spread $[(1-\la)S_t, S_t].$
 The proof of the claim \eqref{int1D} is formalized in the subsequent Lemma \ref{Lemm6.7} below.
 \vspace{4mm}
 
 The process $\widehat{V}_t=\widehat{\varphi}^0_t\widehat{Z}^0_t +\widehat{\varphi}^1_t\widehat{Z}^1_t$ is a uniformly integrable martingale by Theorem \ref{thhurt}.
 By It\^o's lemma and using the fact that $\widehat{\varphi}$ is of finite variation, we have
 \begin{equation*}
   d\widehat{V}_t=\widehat{\varphi}^0_{t-}(d\widehat{M}^0_t-d\widehat{A}^0_t)+ \widehat{\varphi}^1_{t-} (d\widehat{M}^1_t -d\widehat{A}^1_t) 
                  +\widehat{Z}^0_{t} d\widehat{\varphi}^0_t +\widehat{Z}^1_{t} d\widehat{\varphi}^1_t.
 \end{equation*}
 Hence we may write the process $\widehat{V}_t$ as the sum of three integrals 
 \begin{align*}
   \widehat{V}_t =&\int_0^t\left(\widehat{Z}^0_{u} d\widehat{\varphi}^0_u +\widehat{Z}^1_{u} d\widehat{\varphi}^1_u\right) 
                   + \int_0^t\left(\widehat{\varphi}^0_{u-}d\widehat{M}^0_u+ \widehat{\varphi}^1_{u-}d\widehat{M}^1_u\right)\\
                  &- \int_0^t\left(\widehat{\varphi}^0_{u-}d\widehat{A}^0_u+ \widehat{\varphi}^1_{u-}d\widehat{A}^1_u \right).
 \end{align*}
 The first integral defines a non-increasing process by the self-financing condition \eqref{selffinancing} and the fact that 
  $\frac{\widehat{Z}^1_{u}}{\widehat{Z}^0_{u}}$ takes values in $[(1-\lambda)S_u,S_u]$. 
 The second integral defines a local martingale. 
  
  As regards the third term we claim that 
  \begin{equation} \label{intphidA}
    \int_0^t\left(\widehat{\varphi}^0_{u-}d\widehat{A}^0_u+ \widehat{\varphi}^1_{u-}d\widehat{A}^1_u \right)
  \end{equation}
  defines a non-decreasing process. 
  As $\widehat{V}$ is a martingale, this will imply that the process \eqref{intphidA} vanishes. 

  We deduce from \eqref{int1D} that
  \begin{align*}
      & \int_0^t\left(\widehat{\varphi}^0_{u-}d\widehat{A}^0_u+ \widehat{\varphi}^1_{u-}d\widehat{A}^1_u \right) \\
    = & \int_0^t\left(\widehat{\varphi}^0_{u-}d\widehat{A}^0_u+ \widehat{\varphi}^1_{u-}d\widehat{A}^1_u \right)\mathbf{1}_{\{\widehat{\varphi}^1_{u-}\leq0\}} 
        +\int_0^t\left(\widehat{\varphi}^0_{u-}d\widehat{A}^0_u+ \widehat{\varphi}^1_{u-}d\widehat{A}^1_u \right)\mathbf{1}_{\{\widehat{\varphi}^1_{u-}>0\}} \\
    \geq &  \int_0^t\left(\widehat{\varphi}^0_{u-}-\widehat{\varphi}^1_{u-}S_u\right)\mathbf{1}_{\{\widehat{\varphi}^1_{u-}\leq0\}}d\widehat{A}^0_u  
          + \int_0^t\left(\widehat{\varphi}^0_{u-}+\widehat{\varphi}^1_{u-}(1-\lambda)S_u\right)\mathbf{1}_{\{\widehat{\varphi}^1_{u-}>0\}}d\widehat{A}^0_u \\  
     = & \int_0^t \widehat{V}^{liq}_{u-} d\widehat{A}^0_u.
  \end{align*}
  As we have assumed that the liquidation value process $\widehat{V}^{liq}$ satisfies $\inf_{0\leq t\leq T}\widehat{V}^{liq}_{t}>0$ a.s.~and the process $\widehat{A}^0$ is non-decreasing, 
   the vanishing of the process in \eqref{intphidA} implies that $\widehat{A}^0$ vanishes. 
  By \eqref{claim} the processes $\widehat{A}^0$ and $\widehat{A}^1$ vanish simultaneously.

 Summing up, modulo the (still missing) proof of \eqref{int1D},
  we deduce from the fact that $\widehat{V}$ is a martingale that $\widehat{A}^0$ and $\widehat{A}^1$ vanish. 
 Therefore $\widehat{Z}^0$ and $\widehat{Z}^1$ are local martingales.
\end{proof}

\begin{lemma}\label{Lemm6.7}
In the setting of Proposition \ref{proJ37}, let $\widehat{A}^0,\widehat{A}^1$ be the bounded, predictable processes in \eqref{fo} and \eqref{po},
and let $0\le \sigma\le T$ be a stopping time. 
For $\ve >0$, define 
\begin{equation}\label{spar}
  \tau_\ve :=\inf \left\{t\geq\sigma :\frac{S_t}{S_\sigma} =1+\ve \ \mbox{or} \ 1-\ve\right\}.
\end{equation}
Then
\begin{align}\label{G0}
 (1-\ve)(1-\la) S_\sigma \E\Big[\widehat{A}^0_{\tau_\ve} -\widehat{A}^0_\sigma \Big|\F_\sigma\Big] 
   &\le \E\Big[\widehat{A}^1_{\tau_\ve} -\widehat{A}^1_\sigma \Big|\F_\sigma \Big] \\
   & \le (1+\ve)S_\sigma\E\Big[\widehat{A}^0_{\tau_\ve} -\widehat{A}^0_\sigma \Big|\F_\sigma \Big]. \nonumber
\end{align}
\end{lemma}

Before starting the proof, we remark that it is routine to deduce \eqref{int1D} from the lemma.
\begin{proof}
 
 The processes $\widehat{A}^0$ and $\widehat{A}^1$ are c\`adl\`ag, being defined as the differences of two c\`adl\`ag processes. 
 Hence, we have 
\begin{equation*}
 \E\Big[\widehat{A}^1_{\tau_\ve} - \widehat{A}^1_\sigma \Big| \mathcal{F}_\sigma\Big]
 =\lim_{\delta\searrow 0} \E \Big[\widehat{A}^1_{\tau_{\ve+\delta}}-\widehat{A}^1_{\tau_\delta} \Big| \mathcal{F}_\sigma\Big].
\end{equation*}

Fix the sequence $(Z^n)_{n=1}^{\infty}$ of local martingales as above. 
It follows from \eqref{Ee3a} below that we have for all but countably many $\delta >0$, 
that $(Z^n_{\tau_\delta})^\i_{n=1}$ converges to $\widehat{Z}_{\tau_\delta}$ in probability.
The bottom line is that it will suffice to prove \eqref{G0} under the additional assumption 
that $(Z^n_\sigma)^\i_{n=1}$ and $(Z^n_{\tau_\ve})^\i_{n=1}$ converge to $\widehat{Z}_\sigma$ and $\widehat{Z}_{\tau_\ve}$ in probability and 
 -- after passing once more to a subsequence -- almost surely.

To simplify notation, we drop the subscript $\ve$ from $\tau_\ve$. 
We then have almost surely that 
\begin{equation}\label{G1}
  \lim_{n\to\i} \left( Z^{0,n}_\tau -Z^{0,n}_\sigma\right) = \left(\widehat{Z}^0_\tau -\widehat{Z}^0_\sigma\right) =
      \left(\widehat{M}^0_\tau -\widehat{M}^0_\sigma\right) - \left(\widehat{A}^0_\tau -\widehat{A}^0_\sigma\right),
\end{equation}
and
\begin{equation}\label{G2}
  \lim_{n\to\i} \left( Z^{1,n}_\tau -Z^{1,n}_\sigma\right) =\left(\widehat{Z}^1_\tau -\widehat{Z}^1_\sigma\right) =
      \left(\widehat{M}^1_\tau -\widehat{M}^1_\sigma\right) - \left(\widehat{A}^1_\tau -\widehat{A}^1_\sigma\right).
\end{equation}
We also have that
\begin{equation}\label{G3}
  \lim_{C\to\i} \lim_{n\to\i}\E\left[\left(Z^{0,n}_\tau - Z^{0,n}_\sigma\right) \mathbf{1}_{\{Z^{0,n}_\tau -Z^{0,n}_\sigma \geq C\}}\Big|\F_\sigma\right]
    =\E\left[\widehat{A}^0_\tau-\widehat{A}^0_\sigma\Big|\F_\sigma\right],
\end{equation}
holds true a.s., and similarly
\begin{equation}\label{G4}
  \lim_{C\to\i} \lim_{n\to\i}\E\left[\left(Z^{1,n}_\tau - Z^{1,n}_\sigma\right) \mathbf{1}_{\{Z^{1,n}_\tau -Z^{1,n}_\sigma \geq C\}}\Big|\F_\sigma\right]
    =\E\left[\widehat{A}^1_\tau-\widehat{A}^1_\sigma \Big|\F_\sigma\right].
\end{equation}
Indeed, we have for fixed $C>0$
\begin{align*}
  0 &= \E\left[Z^{0,n}_\tau -Z^{0,n}_\sigma \big|\F_\sigma\right] \notag \\
    &=\E\left[\left(Z^{0,n}_\tau -Z^{0,n}_\sigma\right)\mathbf{1}_{\{Z^{0,n}_\tau -Z^{0,n}_\sigma \geq C\}} \Big|\F_\sigma\right] 
        +\E\left[\left(Z^{0,n}_\tau -Z^{0,n}_\sigma\right)\mathbf{1}_{\{Z^{0,n}_\tau -Z^{0,n}_\sigma < C\}} \Big|\F_\sigma\right]. \notag
\end{align*}
Note that
\begin{align*}
  \lim_{C\to\i} \lim_{n\to\i} \E\left[\left(Z^{0,n}_\tau-Z^{0,n}_\sigma\right)\mathbf{1}_{\{Z^{0,n}_\tau -Z^{0,n}_\sigma < C\}}\Big|\F_\sigma\right]
   &=\E\left[\widehat{Z}^0_\tau -\widehat{Z}^0_\sigma \Big|\F_\sigma\right]\\ 
   &=-\E\left[\widehat{A}^0_\tau-\widehat{A}^0_\sigma \Big|\F_\sigma\right], \notag
\end{align*}
where the last equality follows from \eqref{G1}. We thus have shown \eqref{G3}, and \eqref{G4} follows analogously.

We even obtain from \eqref{G3} and \eqref{G4} that
\begin{equation}\label{G5}
 \lim_{C\to\i} \lim_{n\to\i} \E\left[Z^{0,n}_{\tau}\mathbf{1}_{\{Z^{0,n}_{\tau}\geq C\}}  \Big|\F_{\sigma}\right] =
    \E\left[\widehat{A}^0_\tau-\widehat{A}^0_\sigma \Big|\F_\sigma\right]
\end{equation}
and
\begin{equation}\label{G6}
 \lim_{C\to\i} \lim_{n\to\i} \E\left[Z^{1,n}_{\tau}\mathbf{1}_{\{Z^{1,n}_{\tau}\geq C\}} \Big|\F_{\sigma}\right] =
    \E \left[\widehat{A}^1_\tau-\widehat{A}^1_\sigma \Big|\F_\sigma\right]
\end{equation}
Indeed, the sequence $(Z^{0,n}_\sigma)^\i_{n=1}$ converges a.s.~to $\widehat{Z}^0_\sigma$ so that 
by Egoroff's Theorem it converges uniformly on sets of measure bigger than $1-\delta$. 
As we condition on $\F_{\sigma}$ in \eqref{G3}, we may suppose without loss of generality that 
$(Z^{0,n}_\sigma)^\i_{n=1}$ converges uniformly to to $\widehat{Z}^0_\sigma$.
Therefore the terms involving $Z^{0,n}_\sigma$ in \eqref{G3} disappear in the limit $C\to\i.$

Finally, observe that
\begin{equation*}
 \frac{Z^{1,n}_\tau}{Z^{0,n}_\tau} \in \left[(1-\la)S_\tau,S_\tau\right]\subseteq\left[(1-\ve)(1-\la)S_\sigma,(1+\ve)S_\sigma\right].
\end{equation*}

Conditioning again on $\F_\sigma$, this implies on the one hand
\begin{equation*}
 \lim\limits_{C\to\i} \lim\limits_{n\to\i} \E\left[Z^{1,n}_\tau\mathbf{1}_{\{Z^{0,n}_\tau\geq C\}} \Big|\F_\sigma\right] =
     \E \left[\widehat{A}^1_\tau-\widehat{A}^1_\sigma \Big|\F_\sigma\right],
\end{equation*}
and on the other hand
\begin{align*}
 \frac{\E\big[\widehat{A}^1_\tau-\widehat{A}^1_\sigma\big|\F_\sigma\big]}{\E\big[\widehat{A}^0_\tau-\widehat{A}^0_\sigma\big|\F_\sigma\big]} &=
   \lim\limits_{C\to\i} \lim\limits_{n\to\i}
  \frac{\E\big[Z^{1,n}_\tau\mathbf{1}_{\{Z^{0,n}_\tau\geq C\}}\big|\F_\sigma\big]}
       {\E\big[Z^{0,n}_\tau\mathbf{1}_{\{Z^{0,n}_\tau\geq C\}}\big|\F_\sigma\big]}   \\
   & \in\left[(1-\ve)(1-\la)S_\sigma,(1+\ve)S_\sigma\right],
\end{align*}
  which is assertion \eqref{G0}.
\end{proof}

\section{Two Counterexamples}

In this section, we show that the assumption of $(NUPBR)$ in Theorem \ref{thsorry} cannot be replaced in general by the assumption of the local existence of $\mu$-consistent
price systems, for all $0 <\mu <1$.

\begin{proposition}\label{proC1}
There is a continuous, strictly positive semimartingale $S=(S_t)_{0\le t\le T}$ with the following properties.
\begin{enumerate}[(i)]
 \item $S$ satisfies the stickiness property introduced by Guasoni in \textnormal{\cite{G06}}.  
       Hence, for every $0 <\mu <1$, there is a $\mu$-consistent price system.
 \item For fixed $0<\lambda<1$ and $U(x)=\log (x)$, the value function $u(x)$ in \eqref{J18} is finite 
       so that by Theorem \ref{thhurt} there is a dual optimizer $\widehat{Y}=(\widehat{Y}^0_t, \widehat{Y}^1_t)_{0\leq t\leq T}\in\mathcal{B}$. 
 \item The dual optimizer $\widehat{Y}$ fails to be a local martingale.
\end{enumerate}

In fact, there is no shadow price in the sense of Definition \ref{ShadowPriceDef}, i.e., no semimartingale $(\widetilde{S}_t)_{0\le t\le T}$ 
such that $\widetilde{S}$ takes its values in the bid-ask spread $[(1-\la)S, S]$ and such that equality \eqref{J22} holds true.
\end{proposition}

\begin{remark}
The construction in the proof will yield a \textnormal{non-decreasing} process $S$ which will imply in a rather spectacular way 
that $S$ does not satisfy $(NUPBR)$. 
\end{remark}

We start by outlining the proof in an informal way, postponing the technicalities to the Appendix.  
First note that, for logarithmic utility $U(x)=\log(x)$ the normalized dual optimizer $\frac{\widehat{Y}(y)}{y}$ does not depend on $y>0$; 
we therefore dropped the dual variable $y>0$ in $(ii)$ and $(iii)$ above.

Let $B=(B_t)_{t\geq 0}$ be a standard Brownian motion on some underlying probability space $(\Omega,\mathcal{F},\p)$, starting at $B_0=0$, and let
$\mathbb{F}=(\mathcal{F}_t)_{t\geq 0}$ be the $\p$-augmented filtration generated by $B$. For $w\geq 0$, define the Brownian motion $W^w$ with drift, 
starting at $W^w_0=w$, by
 \begin{equation*}
   W^w_t:=w+B_t-t,\quad t\geq 0.
 \end{equation*}
Define the stopping time
 \begin{equation*}
  \tau^w:=\inf\{t>0\ |\ W^w_t\le 0\}
 \end{equation*}
 and observe that the law of $\tau^w$ is inverse Gaussian with mean $w$ and variance $w$ (see e.g.~\cite[I.9]{RW94}).

For fixed $w>0$, the stock price process $S=S^w$ is defined by
\begin{equation}  \label{Sw}
  S^w_t:=e^{t\wedge \tau^w},  \quad t\geq 0.
\end{equation}

Let us comment on this peculiar definition of a stock price process $S$: 
the price can only {\it move upwards}, as it equals the exponential function up to time $\tau^w$;
from this moment on $S$ remains constant (but never goes down).

It is notationally convenient to let $t$ range in the time interval $[0,\i]$. 
To transform the construction into our usual setting of bounded time intervals $[0,T]$, note that
$\tau^w$ is a.s.~finite so that the deterministic time change $u=\arctan (t)$ defines a process $\overline{S}^w_u=S^w_{\arctan(t)}$ 
which can be continuously extended to all $u\in[0,\tfrac{\pi}{2}]$. 
We prefer not to do this notational change and to let $T=\i$ be the terminal horizon of the process $S=(S_t)_{0\le t\le\i}$ and 
of our optimization problem.

Fix transaction costs $\lambda\in(0,1)$, the utility function $U(x)=\log(x)$, and initial endowment $x=1$. 
We consider the portfolio optimization problem \eqref{J5}, i.e.,
\begin{equation}\label{C3}
 \E [\log (g)]\to \max!, \qquad g\in\mathcal{C}^w.
\end{equation}

The super-script $w$ pertains to the initial value $W^w_0$ of the process $W^w$ and will be dropped if there is no danger of confusion.

We shall verify below that $S$ admits a $\mu$-consistent price system, for all $0<\mu <1$, and that the value \eqref{C3} of the optimization problem is finite.

\medskip

Let us discuss at an intuitive level what the optimal strategy for the $\log$-utility optimizing agent should look like. 
Obviously, she will never want to go short on a stock $S$ which only can go up. 
Rather, she wants to invest substantially into this bonanza. 
For an agent without transaction costs, there is no upper bound for such an investment as there is no downside risk. 
Hence, $S$ allows for an ``unbounded profit with bounded risk" and the utility optimization problem degenerates in this case, i.e., $u(x)\equiv\i$.

More interesting is the situation when the agent is confronted with transaction costs $0 <\la <1$. Starting from initial endowment $x=1$, 
i.e., $(\varphi^0_{0-},\varphi^1_{0-})$ $=(1,0)$,
there is an upper bound for her investment into the stock at time $t=0$, namely $\tfrac{1}{\la}$ shares.

This is the maximal amount of holdings in stock which yields a non-negative liquidation value $V^{liq}_{0}(\varphi)$.
Indeed, in this case $(\varphi^0_{0}, \varphi^1_{0})=(1-\frac{1}{\la},\frac{1}{\la})$ 
implies that $V^{liq}_{0}(\varphi)=1-\tfrac{1}{\la} +(1-\la) \tfrac{1}{\la}=0$.

This gives rise to the following notation.

\begin{definition}
Let $\varphi=(\varphi^0_t,\varphi^1_t)_{0\le t\le \i}$ be a self-financing trading strategy for $S$ such that $\varphi^0_t+\varphi^1_tS_t>0$. 
The \textnormal{leverage} process is defined by
\begin{equation*}
  L_t(\varphi)=\frac{\varphi^1_t S_t}{\varphi^0_t +\varphi^1_t S_t}, \qquad t\geq 0.
\end{equation*}
\end{definition}

The process $L_t(\varphi)$ may be interpreted as the ratio of the value of the position in stock to the total value of the portfolio if we do not consider transaction costs. 
We obtain from the above discussion that the process $L_t(\varphi)$ is bounded by
$\frac{1}{\la}$ if $\varphi$ is admissible, i.e., if
\begin{equation*}
  V_t^{liq}(\varphi) =\varphi^0_t +(1-\la) \varphi^1_t S_t\geq 0, 
\end{equation*}
for $t\geq 0.$

What is the optimal leverage which the log-utility maximizer chooses, say at time $t=0$? 
The answer depends on the initial value $w$ of the process $W^w$. 
If $w$ is very small, it is intuitively rather obvious that the optimal strategy $\widehat{\varphi}$ only uses leverage 
$L_{0}(\widehat{\varphi})=0$ at time $t=0$, i.e., it is optimal to keep all the money in bond. 
Indeed, in this case $\tau^w$ takes small values with high probability. 
If the economic agent decides to buy stock at time $t=0$, then --- due to transaction costs --- she will face a loss with high probability, 
as she has to liquidate the stock before it has substantially risen in value. 
For sufficiently small $w$, these losses will outweigh the gains which can be achieved when $\tau^w$ takes large values. 
Hence, for $w$ sufficiently small, say $0<w \le {\underline{w}}$, we expect that the best strategy is not to buy any stock at time $t=0$.

Now we let the initial value $w$ range above this lower threshold $\underline{w}$. 
As $w$ increases, it again is rather intuitive from an economic point of view that the agent will dare to take an increasingly higher leverage at time $t=0$. 
Indeed, the stopping times $\tau^w$ are increasing in $w$ so the prospects for a substantial rise of the stock price become better as $w$ increases.

The crucial feature of the example is that we will show that there is a finite upper threshold $\overline{w} >0$ such that, for $w \geq\overline{w}$, 
the optimal strategy $\widehat{\varphi}$ at time $t=0$ takes {\it maximal leverage}, i.e., $L_{0}(\widehat{\varphi})=\tfrac{1}{\la}$. 
In fact, the optimal strategy $\widehat{\varphi}$ will then satisfy $L_t(\widehat{\varphi})=\tfrac{1}{\la}$ and therefore $V^{liq}_t(\widehat{\varphi})=0$
 as long as $W^w_t$ remains above the threshold $\overline{w}$.

\begin{lemma} \label{LeverageLemma}
Using the above notation there is $\overline{w} >0$ such that, for $w\geq \overline{w}$, the optimizer $\widehat{\varphi}^w$ of the optimization problem \eqref{C3}
satisfies
\begin{equation*}
  L_{0} (\widehat{\varphi}^w)=\frac{1}{\la}.
\end{equation*}

More precisely, fix $w =\overline{w} +1$, and define $\sigma :=\inf \{t>0\,|\,W^w_t \le \overline{w} \}$.
Then
\begin{equation}\label{young}
L_t(\widehat{\varphi}^w) =\frac{1}{\la},
\end{equation}
for $0\leq t\le \sigma$.

For $0\leq t\leq\sigma$ we may then explicitly calculate the primal optimizer
\begin{equation}
  \widehat{\varphi}^0_t =\big(1-\tfrac{1}{\lambda}\big)\exp\big(\tfrac{1}{\lambda}t\big), \qquad 
  \widehat{\varphi}^1_t = \tfrac{1}{\lambda}\exp\Big(\big(\tfrac{1}{\lambda}-1\big)t\Big),
\end{equation}
and the dual optimizer
\begin{equation}
 \widehat{Y}^0_t = \exp\big(-\tfrac{1}{\lambda}t\big), \qquad \widehat{Y}^1_t = \exp\Big(\big(1-\tfrac{1}{\lambda}\big)t\Big),
\end{equation}
so that
 \begin{equation} \label{ShadowCandidate}
   \widehat{S}_t:=\frac{\widehat{Y}^1_t}{\widehat{Y}^0_t} = S_t, 
 \end{equation}
for $0\le t\le\sigma$.
\end{lemma}

Admitting this lemma, whose proof is given in the Appendix, we can quickly show Proposition \ref{proC1}. 
The crucial assertion is that there is no shadow price $\widetilde{S}$. 

\begin{proof}[Proof of Proposition \ref{proC1}]
 Using the above notation fix $w=\overline{w}+1$ and drop this super-script to simplify notation.

 \vspace{3mm}
 
 $(i)$  
 We claim that the process $S$ has the {\it stickiness property} as defined by P.~Guasoni \cite[Definition 2.2]{G06}: 
 this property states that, for any $\ve>0$ and any stopping time $\sigma$ with $\p[\sigma<\infty]>0$, we have, conditionally on $\{\sigma <\infty\}$, 
  that the set of paths $(S_t)_{t\geq\sigma}$, which do not leave the price corridor $[\frac{1}{1+\ve}S_\sigma,(1+\ve)S_\sigma]$, has strictly positive measure. 
 Combining \cite[Corollary 2.1]{G06} and \cite[Theorem 2]{GRS10}, we have that, 
  for the continuous process $S$, the stickiness property of $Y$ implies that $S$ verifies $(CPS^{\mu})$ for all $0<\mu<1$.
  
 To show the stickiness property simply observe that, for each $\delta>0$ and each stopping time $\sigma$ such that $\p[\sigma<\tau]>0$, we have
 $$ \p[\sigma<\tau,~|\tau-\sigma|<\delta]>0. $$
 Indeed, given $\sigma$ such that $\p[\sigma<\tau]>0$, i.e., $W$ has not yet reached zero at time $\sigma$, 
  $(W_t)_{t\geq\sigma}$ will hit zero with positive probability before more than $\delta$ units of time elapse.
  
 \vspace{3mm} 
 
 $(ii)$ For fixed $0<\lambda<1$ and $\varphi\in\mathcal{A}(x)$, we observe that $V_{\tau}^{liq}(\varphi)\leq x\exp\big(\tfrac{\tau}{\lambda}\big)$.
 As $\tau$ has expectation $\E[\tau]=w$, we obtain that 
  $$ u(x)\leq \E\left[\log\Big(x\exp\big(\tfrac{\tau}{\lambda}\big)\Big)\right]=\log(x)+\frac{1}{\lambda}\E[\tau] = \log(x)+\frac{w}{\lambda}<\infty. $$
 Hence by Theorem \ref{thhurt} there is a dual optimizer $\widehat{Y}\in\mathcal{B}$.  
 
 \vspace{3mm}
 
 $(iii)$ Lemma \ref{LeverageLemma} provides very explicitly the form of the primal and dual optimizer $\widehat{\varphi}$ and $\widehat{Y}$ respectively, for $0\leq t\leq \sigma$.
 In particular, $\widehat{Y}$ is a supermartingale, which fails to be a local martingale. 
 
 \vspace{3mm}
 
 We now turn to the final assertion of Proposition \ref{proC1}. 
 We know from \cite[Theorem 3.6]{CS14} that the process $\widehat{S}_t:=\frac{\widehat{Y}^1_t}{\widehat{Y}^0_t}$ is a shadow price process in the generalized sense of \cite[Theorem 3.6]{CS14}. 
 By \eqref{ShadowCandidate} we have $\widehat{S}_t=S_t$, for $0\leq t\leq\sigma$. 
 
 Let us recall this {\it generalized sense} of a shadow price as stated in \cite[Theorem 3.6]{CS14}: 
 for every competing finite variation, self-financing trading strategy $\varphi\in\mathcal{A}(x)$ 
 such that the {\it liquidation value} remains non-negative, i.e., 
 \begin{equation} \label{F1a}
   V_t^{liq}(\varphi) = \varphi^0_t+(\varphi_t^1)^+(1-\la)S_t-(\varphi_t^1)^-S_t\geq 0, 
 \end{equation}
 for all $0\leq t\leq T$, we have 
 \begin{equation} \label{F1}
   \E\left[U\left(x+\big(\varphi^1\sint\widehat{S}\big)_T\right)\right] \leq \E\left[U\left(V^{liq}_T(\widehat{\varphi})\right)\right]=u(x).
 \end{equation}

 This generalized shadow price property does hold true for the above process $\widehat{S}:=\frac{\widehat{Y}^1}{\widehat{Y}^0}$ by \cite[Theorem 3.6]{CS14}. 
 In fact, as everything is very explicit in the present example, at least for $0\leq t\leq\sigma$, this also can easily be verified directly. 
 
 But presently, we are considering the shadow price property in the more classical sense of Definition \ref{ShadowPriceDef}, 
 where we allow $\varphi^1$ in \eqref{F1} to range over all predictable $\widehat{S}$-integrable processes which are admissible only in the sense
 \begin{equation} \label{F2}
   x+\big(\varphi^1\sint\widehat{S}\big)_t \geq 0,
 \end{equation}
 for all $0\leq t\leq T$.
 This condition is much weaker than \eqref{F1a}. 
 
 Clearly $\widehat{S}$ is the only candidate for a shadow price process in the sense of Definition \ref{ShadowPriceDef}. 
 But as $\widehat{S}$ only moves upwards, for $0\leq t\leq\sigma$, it can certainly not satisfy this property. 
 Indeed, the left-hand side of \eqref{J22} must be infinity:
 \begin{equation}\label{P2}
  \sup\limits_{\widetilde{g}\in \widetilde{\mathcal{C}}(x)} \E [U(\widetilde{g})]=\i.
 \end{equation}

For example, it suffices to consider the integrands $\varphi^1= C\mathbf{1}_{\rrbracket 0,\tau\rrbracket}$ to obtain 
 $(\varphi^1 \sint \widehat{S})_T= C (S_\sigma -S_0) = C(e^{\sigma}-1)$.
Sending $C$ to infinity we obtain \eqref{P2}.

 This shows that there cannot be a shadow price process $\widetilde{S}$ as in Definition \ref{ShadowPriceDef}.
\end{proof}

\begin{remark}
As pointed out by one of the referees the construction of Proposition \ref{proC1} uses the Brownian filtration $\mathbb{F}^B=(\F^B_t)_{t\geq 0}$ 
  while the natural ($\p$-augmented) filtration $\mathbb{F}^S=(\F^{S}_t)_{t\geq 0}$ generated by the price process $S=(S_t)_{t\geq 0}$ is much smaller.
The referee raised the question whether this discrepancy of the filtrations may be avoided. Fortunately, the anwser is yes. 
Let us define 
  $$ G_t := \int_0^t W_u^w du, \quad t\geq 0, $$
  where $(W_u^w)_{u\geq 0}$ is defined above. 
Note that the process $G$ generates the Brownian filtration $\mathbb{F}^B$. 
Also note that the process $G$ is increasing up to time $\tau$. 
We may replace the process $S_t = e^{t\wedge\tau}$ in the construction of Proposition \ref{proC1} by the process 
  $$ \overline{S}_t := \exp\big((t+G_t)\wedge(\tau+G_{\tau})\big). $$
The reader may verify that $\overline{S}$ also has the properties claimed in Proposition \ref{proC1} and has the additional feature to generate the Brownian filtration up to time $\tau$. 
\end{remark}

 
We finish this section by considering a variant of the example constructed in Proposition \ref{proC1}. 
The main feature of this modified example is to show that, for a continuous process $S$, it may happen that 
 $\widehat{S}:=\tfrac{\widehat{Y}^1}{\widehat{Y}^0}$ is a shadow price in the sense of Definition \ref{ShadowPriceDef}, but fails to be continuous. 

\vspace{2mm}

Consider the first jump time $\tau^{\alpha}$ of a Poisson process $(N_t)_{t\geq 0}$ with parameter $\alpha>0$. 
It is exponentially distributed with parameter $\alpha>0$, so that $\E[\tau^{\alpha}]=\alpha^{-1}$. 
The stock price process $S=S^{\alpha}$ is defined by 
 \begin{equation*}
   S^{\alpha}_t:=e^{t\wedge\tau^{\alpha}}. 
 \end{equation*}
Similarly, as in the previous example, the price moves upwards up to time $\tau^{\alpha}$, and then remains constant. As information available to the investor we use the $\p$-augmented filtration $\mathbb{F}^{S^\alpha}=(\F^{S^\alpha}_t)_{t\geq 0}$ generated by the price process $S^\alpha=(S^\alpha_t)_{t\geq 0}$.

 For fixed transaction costs $\lambda\in(0,1)$ such that $\lambda<\alpha^{-1}$, and initial endowment $x>0$, 
 we consider the portfolio optimization problem \eqref{J5} with logarithmic utility function, i.e., 
 \begin{equation}\label{ProblemExample2}
  \E\left[\log\left(g\right)\right] \to \max!, \qquad g\in\mathcal{C}(x).
 \end{equation}

\begin{proposition} \label{CounterExProp}
 The process $S^{\alpha}$ has the following properties. 
 \begin{enumerate}[(i)]
  \item The price process $S^{\alpha}$ satisfies the condition ($CPS^{\mu}$) for all $\mu\in(0,1)$, but does not satisfy the condition $(NUPBR)$. 
  \item The value function $u(x)$ is finite, for $x>0$.
  \item The dual optimizer $\widehat{Y}\in\mathcal{B}$ is induced by a martingale $\widehat{Z}$ and therefore Theorem \ref{shadow} implies that
         $ \widehat{S}=\frac{\widehat{Z}^1}{\widehat{Z}^0} $ is a shadow price in the sense of Definition \ref{ShadowPriceDef}. 
  \item The shadow price $\widehat{S}$ fails to be continuous. In fact it has a jump at time $t=\tau^{\alpha}$.
 \end{enumerate}
\end{proposition}
 
 Again, we start by arguing heuristically to derive candidates for primal and dual optimizer. 
 Then we shall verify, using the duality theorem, 
  that they are actually optimizers to the primal and dual problem respectively.    

 Since $S^{\alpha}$ can never move downwards, it is rather intuitive that the agent will never go short on this (see Proposition \ref{prop:opt} for a formal argument), 
  hence the leverage process is always positive, 
 i.e., 
 $$ L_t(\varphi)=\frac{\varphi^1_tS_t}{\varphi^0_t+\varphi^1_tS_t}\geq 0. $$ 
 By the memorylessness of the exponential distribution and the properties of $U(x)=\log(x)$, 
 the optimal leverage should remain constant on the stochastic time interval $\llbracket 0,\tau^{\alpha}\rrbracket$.
 Under transaction costs $\lambda>0$, the upper bound for the leverage $L_t(\vp)$ is $\frac{1}{\lambda}$ as above.
 
 Fix $\ell\in[0,\frac{1}{\lambda}]$. Starting with initial endowment $(\varphi^0_{0-},\varphi^1_{0-})=(x,0)$, we buy $\ell x$ shares, 
 i.e., $(\varphi_{0}^0,\varphi_{0}^1)=((1-\ell)x,\ell x).$

 Similarly, as above, we derive from the requirement that $L_t(\vp)$ remains constant that 
 $$ (\varphi^0_t,\varphi^1_t) = \left((1-\ell)xe^{\ell t}, \ell xe^{(\ell-1)t}\right), $$
  for $ 0\leq t\leq \tau^{\alpha}$, 
  which yields that the liquidation value at time $\tau^{\alpha}$ is
  $$V_{\tau^{\alpha}}^{liq}(\varphi^0,\varphi^1)=(1-\ell\lambda)xe^{\ell\tau^{\alpha}}.$$
 Hence the expected utility is
 \begin{align*}
     f(\ell) &:= \E\left[\log\big(V_{\tau^{\alpha}}(\varphi^0,\varphi^1)\big)\right] = \E\left[\log\left((1-\ell\lambda)xe^{\ell\tau^{\alpha}}\right)\right] \\
          &\,\,= \log(x)+\log(1-\ell\lambda)+\tfrac{\ell}{\alpha}. 
 \end{align*}
 Maximizing over $\ell\in[0,\frac{1}{\lambda}]$, we get the optimal leverage $$\hat{\ell}=\frac{1-\alpha\lambda}{\lambda}\vee 0.$$
 Therefore the educated guess for the optimal strategy is  
 \begin{align*}
  (\widehat{\varphi}_t^0,\widehat{\varphi}_t^1) &= \left((1-\hat{\ell})xe^{\hat{\ell}t},\hat{\ell}xe^{(\hat{\ell}-1)t}\right)  \\
     &= \Big(\tfrac{\lambda-1+\alpha\lambda}{\lambda}x\exp\left(\tfrac{1-\alpha\lambda}{\lambda}t\right),
               \tfrac{1-\alpha\lambda}{\lambda}x\exp\left(\tfrac{1-\alpha\lambda-\lambda}{\lambda}t\right)\Big),  
 \end{align*}
  for $0\leq t < \tau^{\alpha}$. 
 At $\tau^{\alpha}$ the portfolio may be liquidated so that $(\widehat{\varphi}_t^0,\widehat{\varphi}_t^1)=\big(V^{liq}_{\tau^{\alpha}}(\widehat{\varphi}),0\big)$ for $t\geq\tau^{\alpha}$.
 This yields as candidate for the value function $u(x)$
  \begin{align*}
      \bar{u}(x) = \log(x)+\log(1-\hat{\ell}\lambda) + \tfrac{\hat{\ell}}{\alpha} 
                 = \log(x) + \log(\alpha\lambda) + \tfrac{1-\alpha\lambda}{\alpha\lambda},
  \end{align*}
  which satisfies $\bar{u}(x)\leq u(x)$. 
  
 \vspace{4mm}

 Let us continue our heuristic search for the dual optimizer $\widehat{Z}$ and the shadow price $\widehat{S}=\frac{\widehat{Z}^1}{\widehat{Z}^0}$.
  
  As for a Poisson process $(N_t)_{t\geq 0}$ and $u<1$, the process 
  $$ \exp\big(\log(1-u)N_t+u\alpha t\big), \quad t\geq 0 $$
  is a martingale, we use the following ansatz to look for the dual optimizer, where $u,v<1$ are still free variables. 
  
  Set 
  \begin{eqnarray*}
    && Z^0_{\tau^{\alpha}\wedge t} := \exp\big(\log(1-u)N_{\tau^{\alpha}\wedge t}+u\alpha(\tau^{\alpha}\wedge t)\big),  \\
    && Z^1_{\tau^{\alpha}\wedge t} := \exp\big(\log(1-v)N_{\tau^{\alpha}\wedge t}+v\alpha(\tau^{\alpha}\wedge t)\big),  \\ 
    &&\widetilde{S}_t:=\tfrac{Z^1_{\tau^{\alpha}\wedge t}}{Z^0_{\tau^{\alpha}\wedge t}} 
                     = \exp\Big(N_{\tau^{\alpha}\wedge t}\log\left(\tfrac{1-v}{1-u}\right)\Big)\exp\big((v-u)\alpha(\tau^{\alpha}\wedge t)\big).
  \end{eqnarray*}
  By the definition of $\tau^{\alpha}$, we have 
  \begin{equation*}
     \widetilde{S}_t = \left\{
                      \begin{array}{ll}
                          \exp((v-u)\alpha t),                        & \textnormal{ if }0\leq t< \tau^{\alpha}, \\
                          \frac{1-v}{1-u}\exp\big((v-u)\alpha\tau^{\alpha}\big), & \textnormal{ if } t\geq \tau^{\alpha}.
                       \end{array}
                       \right.
  \end{equation*}
  In order to be a candidate for a shadow price, $\widetilde{S}$ should satisfy
  \begin{equation*}
     \widetilde{S}_t = \left\{
                      \begin{array}{ll}
                                    S_t,                                & \textnormal{ if }0\leq t< \tau^{\alpha}, \\
                          (1-\lambda)S_{\tau^{\alpha}}, & \textnormal{ if } t\geq \tau^{\alpha},
                       \end{array}
                       \right. 
  \end{equation*}
  therefore the parameters $u$ and $v$ should solve the following equations
  $$ v-u=\frac{1}{\alpha},\quad \frac{1-v}{1-u} = 1-\lambda. $$
  Solving the equations above, we obtain $u=1-\frac{1}{\alpha\lambda}$ and $v=1+\frac{1}{\alpha}-\frac{1}{\alpha\lambda}$ so that
  \begin{eqnarray*}
    && \widehat{Z}^0_{t} := \left(\tfrac{1}{\alpha\lambda}\right)^{N_{\tau^{\alpha}\wedge t}}
                                            \exp\left(\left(\alpha-\tfrac{1}{\lambda}\right)(\tau^{\alpha}\wedge t)\right), \\
    && \widehat{Z}^1_{t} := \left(\tfrac{1}{\alpha\lambda}-\tfrac{1}{\alpha}\right)^{N_{\tau^{\alpha}\wedge t}}
                                            \exp\left(\left(1+\alpha-\tfrac{1}{\lambda}\right)(\tau^{\alpha}\wedge t)\right), \\
    &&\, \widehat{S}_t := \tfrac{\widehat{Z}^1_t}{\widehat{Z}^0_t} 
                      = (1-\lambda)^{N_{\tau^{\alpha}\wedge t}}e^{\tau^{\alpha}\wedge t}.
  \end{eqnarray*}
  
 This finishes our heuristic considerations. 
 We shall now apply duality theory to verify the above guesses. 
  
  \begin{proof}[Proof of Proposition \ref{CounterExProp}]
  Assertions $(i)$ and $(ii)$ follow by the same token as in Proposition \ref{proC1} above. 
  
  \vspace{2mm}
  
  As regards $(iii)$ and $(iv)$ note that $(\widehat{Z}^0_t,\widehat{Z}^1_t)_{t\geq 0}$ is $\p$-martingale. 
  As $(\widehat{Z}^0_t,\widehat{Z}^1_t)_{t\geq 0}$ is strictly positive and satisfies 
  $$ (1-\lambda)S_t\leq\frac{\widehat{Z}^1_t}{\widehat{Z}^0_t}\leq S_t, $$ 
  for all $t\geq 0$, it defines a $\lambda$-consistent price system.
  
  For $\widehat{y}:=\bar{u}'(x)=\frac{1}{x}$, we have 
  \begin{align*}
   v(\widehat{y}) & \leq  \E\Big[-\log\big(\widehat{y}\widehat{Z}^0_{\tau^{\alpha}}\big)-1\Big] 
              = -\E\left[\log\left(\tfrac{1}{x}\tfrac{1}{\alpha\lambda}e^{\left(\alpha-\tfrac{1}{\lambda}\right)\tau^{\alpha}}\right)\right]-1 \\
         & = \log(x)+\log(\alpha\lambda)+\tfrac{\alpha\lambda-1}{\alpha\lambda}-1 = \bar{u}(x)-x \widehat{y} \\
         & \leq u(x)-x\widehat{y}.
  \end{align*}
  Combining this inequality with the trivial Fenchel inequality $v(\widehat{y})\geq u(x)-x\widehat{y}$ we obtain
   $u(x)-x\widehat{y} = \bar{u}(x)-x \widehat{y} = v(\widehat{y})$, in particular $u(x)=\bar{u}(x)$.  
  From Theorem \ref{thhurt}, $(\widehat{\varphi}_t^0,\widehat{\varphi}_t^1)_{t\geq 0}$ is 
    indeed an optimal strategy of the problem defined in \eqref{ProblemExample2},
    and $(\widehat{Z}^0_t,\widehat{Z}^1_t)_{t\geq 0}$ is a dual optimizer, which is a $\p$-martingale. 
  According to Theorem \ref{shadow}, it follows that $\widehat{S}$ is a shadow price. 
\end{proof}

\begin{appendix}
 
\section{Proofs and technical results}

\begin{proof}[Proof of Proposition \ref{queen}]
 Let $Z=(Z^0_t,Z^1_t)_{0\leq t\leq T}$ be a $\lambda$-consistent local martingale deflator. 
 By definition, there exists a localizing sequence $(\tau_n)_{n=1}^{\infty}$ of stopping times, 
   so that $Z^{\tau_n}$ is in $\mathcal{Z}^e(S^{\tau_n})$, 
   i.e., $Z^{\tau_n}=(Z^0_{t\wedge\tau_n},Z^1_{t\wedge\tau_n})_{0\leq t\leq T}$ is a local martingale and 
   \begin{equation} \label{Z(1-l)S<Z<ZS}
     (1-\lambda)S_{t\wedge\tau_n} Z^0_{t\wedge\tau_n} \leq Z^1_{t\wedge\tau_n} \leq S_{t\wedge\tau_n} Z^0_{t\wedge\tau_n}.
   \end{equation}
   
 Let $(\varphi^0_t,\varphi^1_t)_{0\leq t\leq T}$ be an admissible self-financing trading strategy under transaction costs $\lambda$. 
 Consider now the stopped value process 
   $$ V^{\tau_n}_t:= \varphi^0_{t\wedge\tau_n}Z^0_{t\wedge\tau_n} + \varphi^1_{t\wedge\tau_n}Z^1_{t\wedge\tau_n}. $$
 Using It\^o's lemma, we obtain 
   $$ dV^{\tau_n}_t = \varphi^0_{t\wedge\tau_n-}dZ^0_{t\wedge\tau_n} + \varphi^1_{t\wedge\tau_n-}dZ^1_{t\wedge\tau_n} 
                      + Z^0_{t\wedge\tau_n}d\varphi^0_{t\wedge\tau_n} + Z^1_{t\wedge\tau_n}d\varphi^1_{t\wedge\tau_n},\quad a.s.  $$
 By \eqref{selffinancing} and \eqref{Z(1-l)S<Z<ZS} we obtain 
   $$ Z^0_{t\wedge\tau_n}d\varphi^0_{t\wedge\tau_n} + Z^1_{t\wedge\tau_n}d\varphi^1_{t\wedge\tau_n} \leq 0, \quad a.s. $$
 As $(Z^0_{t\wedge\tau_n},Z^1_{t\wedge\tau_n})_{0\leq t\leq T}$ is a local martingale, $V^{\tau_n}$ is a local supermartingale.
 As $V^{\tau_n}$ is non-negative, it is a supermartingale, therefore $(V_t)_{0\leq t\leq T}$ is a local supermartingale.
 Again by non-negativity, $(V_t)_{0\leq t\leq T}$ is a supermartingale. Therefore the assertion follows. 
\end{proof}

 In the proof of Lemma \ref{Lemm6.7} we have used the following consequence of the Fatou-limit construction of F\"ollmer and Kramkov \cite[Lemma 5.2]{FK97}. 
 (Compare also \cite[Proposition 2.3]{CS13} for a more refined result.)
 
\begin{lemma} \label{attersee}
  Let $(Z^n)_{n=1}^{\infty}$ be a sequence of $[0,\infty)$-valued $($c\`adl\`ag$)$ supermartingales $Z^n= (Z^n_t)_{0\leq t\leq T}$, all starting at $Z_0^n=1$. 
  There exists a sequence of forward convex combinations, still denoted by $(Z^n)_{n=1}^{\infty}$, a limiting (c\`adl\`ag) supermartingale $Z$ 
  as well as a sequence $(\tau_n)_{n=1}^{\infty}$ of stopping times 
  such that, for every stopping time $0\leq\tau\leq T$ with $\p[\tau=\tau_n]=0$, for each $n\in\N$, we have 
 \begin{equation}\label{Ee3a}
   Z_{\tau} = \p-\lim_{n\to\infty} Z_{\tau}^n,
 \end{equation}
 the convergence holding true in probability. 
\end{lemma}

\begin{proof}
 In \cite[Theorem 2.7]{CS13}, it is shown that there exists a (l\`adl\`ag) optional strong supermartingale $\overline{Z}=(\overline{Z}_t)_{0\leq t\leq T}$ such that, 
 after passing to forward convex combinations of $(Z^n)_{n=1}^{\infty}$, we have 
 \begin{equation}\label{Ee3b}
   \overline{Z}_{\tau} =\p-\lim_{n\to\infty} Z^n_{\tau}, 
 \end{equation}
 for all stopping times $0\leq \tau\leq T$. 
 We shall see that the c\`adl\`ag version of $\overline{Z}$ then is our desired supermartingale $Z$.
 We note in passing that $Z$ is the Fatou-limit of $(Z^n)_{n=1}^{\infty}$ as constructed by F\"ollmer and Kramkov in \cite{FK97}. 
 
 Indeed, we may find a sequence $(\tau_n)_{n=1}^{\infty}$ of stopping times exhausting all the jumps of $\overline{Z}$. 
 Therefore for a stopping time $\tau$ avoiding all the $\tau_n$, we have $Z_{\tau}=\overline{Z}_{\tau}$ so that in this case \eqref{Ee3b} implies \eqref{Ee3a}.
\end{proof}

\subsection{Proof of Lemma \ref{LeverageLemma}}

 Consider the price process $(S^{w}_t)_{t\geq 0}$ as in \eqref{Sw}. 
 Fix proportional transaction costs $\lambda\in(0,1)$ as well as real numbers $\varphi_0^0$ and $\varphi_0^1$. 
 We consider the problem
 \begin{equation} \label{P1}
   \E\big[\log\big(V^{liq}_{\tau^w}(\vp^0,\vp^1)\big)\big]\to\max!,\qquad (\vp^0,\vp^1)\in\cA^w(\vp^0_{0},\vp^1_{0}),
 \end{equation}
 where $\cA^w(\vp^0_0,\vp^1_0)$ denotes the set of all self-financing and admissible trading strategies $(\vp^0,\vp^1)$ under transaction costs $\lambda$ starting with initial endowment $(\vp^0_{0},\vp^1_{0})$. 
 If we do not need the dependence on $w$ explicitly, we drop the superscript $w$ in the sequel to lighten the notation and simply write $W$, $\tau$, $S$ and $\cA(\vp^0_0,\vp^1_0)$.
 
 \begin{proposition}\label{prop:exist}
   Fix $w\geq 0$.  
   For all $(\vp^0_{0},\vp^1_{0})$ with $V_0^{liq}(\vp^0_{0},\vp^1_{0})>0$, there exists an optimal strategy $\hvp=(\hvp^0_t,\hvp^1_t)_{0\leq t<\infty}$ to problem \eqref{P1}
   and we have that
   \begin{align*}
    u(\vp^0_0,\vp^1_0):={}&\sup_{(\vp^0,\vp^1)\in\cA(\vp^0_0,\vp^1_0)}\E\big[\log\big(V^{liq}_{\tau}(\vp^0,\vp^1)\big)\big]\\
                       ={}&\inf_{y>0}\left\{\inf_{(Z^0,Z^1)\in\mathcal{Z}^\lambda}\left\{\E\big[-\log(yZ^0_T)-1\big]+y\E[Z^0_0\vp^0_0+Z^1_0\vp^1_0]\right\}\right\},
   \end{align*}
   where $\mathcal{Z}^\lambda$ denotes the set of $\lambda$-consistent price systems.
 \end{proposition}
 
 \begin{proof}
  Since $U(x)=\log(x)$ has reasonable asymptotic elasticity, $S=(S_t)_{0\leq t<\infty}$ satisfies the condition $(CPS^{\mu})$ for all $\mu\in(0,1)$ by Proposition \ref{proC1}.(i), 
  the assertions follow from the general static duality results for utility maximization under transaction costs as soon as we have shown that $u(\vp^0_0,\vp^1_0)<\infty$; 
  compare \cite{DPT01} and Section 3.2 in \cite{BM03}.
  
  For the latter, we observe that 
    \begin{equation*}
       V^{liq}_\tau(\vp^0,\vp^1)\leq (\vp^0_0+\vp^1_0)\exp(\tfrac{1}{\lambda}\tau)
    \end{equation*}
   and $\tau$ has an inverse Gaussian distribution with mean $\E[\tau]=w$, 
   which implies 
     \begin{equation*}
       u(\vp^0_0,\vp^1_0) \leq \log(\vp^0_0+\vp^1_0)+\tfrac{1}{\lambda}\E[\tau] = \log(\vp^0_0+\vp^1_0)+\tfrac{1}{\lambda}w < \infty,
     \end{equation*}
   hence the proof is completed.
 \end{proof}

 In order to show Lemma \ref{LeverageLemma} we define the value function $v(l,w)$ on $[0,\frac{1}{\lambda}]\times [0,\infty)$ by
  \begin{equation*}
    v(l,w):=\sup_{(\vp^0,\vp^1)\in\cA^w(1-l,l)}\E\big[\log\big(V^{liq}_{\tau^w}(\vp^0,\vp^1)\big)\big],
  \end{equation*}
  where $(\vp^0,\vp^1)\in\cA^w(1-l,l)$ ranges through all admissible trading strategies starting at $(\vp^0_{0-},\vp^1_{0-})=(1-l,l)$. 
 We shall see that, for fixed $w$, the function $v(l,w)$ is {\it decreasing} in $l$: indeed, one may always move at time $t=0$ to a higher degree of leverage; 
 but not vice versa, in view of the transaction costs $\lambda$.

 \begin{lemma} \label{l:propv}
  For fixed $0<\lambda<1$. The value function $v:[0,\frac{1}{\lambda}]\times [0,\infty)\to\R\cup\{-\infty\}$ has the following properties:
  \begin{enumerate}[(1)]
   \item $v(l,w)$ is concave and non-increasing in $l$ for all $w\in[0,\infty)$ and $v(l,0)=\log(1-\lambda l)$.
   \item $v(l,w)$ is non-decreasing in $w$ for all $l\in[0,\frac{1}{\lambda}]$.
   \item $v$ is jointly continuous and $v(l,w)=-\infty$ if and only if $(l,w)=(\frac{1}{\lambda},0)$.\footnote{With continuity at $-\infty$ defined in the usual way.}
   \item $v$ satisfies the \emph{dynamic programming principle}, i.e.,
             $$v(l,w)=\sup_{(\vp^0,\vp^1)\in\cA^w(1-l,l)}\E\left[\log\left(\vp^0_{\tau^w\wedge\sigma}+\vp^1_{\tau^w\wedge\sigma}S^w_{\tau^w\wedge\sigma}\right)
                       +v\big(L_{\tau^w\wedge\sigma}(\varphi),W^w_{\tau^w\wedge\sigma}\big)\right]$$
          for all stopping times $\sigma$.
   \item There exists a non-decreasing, c\`adl\`ag function $\ell:[0,\infty)\to[0,\frac{1}{\lambda}]$ given by
          \begin{equation}\label{def:ell}
            \ell(w):=\textstyle\max\left\{l\in\left[0,\frac{1}{\lambda}\right]~\big|~v(l,w)=v(0,w)\right\}
          \end{equation}
         such that
         \begin{enumerate}[(i)]
          \item $v(l,w)=\max_{k\in[0,\frac{1}{\lambda}]} v(k,w)$ for all $l\in[0,\ell(w)]$.
          \item $v(l,w)$ is strictly concave and strictly decreasing in $l$ on $\big(\ell(w),\frac{1}{\lambda}\big]$.
         \end{enumerate}
  \end{enumerate}
 \end{lemma}

 \begin{proof}
  (1) As $$\cA\Big(1-\big(\mu l_1+(1-\mu)l_2\big), \big(\mu l_1+(1-\mu)l_2\big)\Big)\subseteq \mu\cA(1-l_1,l_1)+(1-\mu)\cA(1-l_2,l_2)$$
      for all $l_1,l_2\in[0,\frac{1}{\lambda}]$ and $\mu\in[0,1]$, 
      the concavity of $v(l,w)$ in $l$ follows immediately from that of $\log(x)$ and $V_\tau(\vp^0,\vp^1)$, as $\log(x)$ is non-decreasing.

     If $l_1<l_2$, the investor with initial endowment $(\vp^0_0,\vp^1_0)=(1-l_1,l_1)$ can immediately buy $(l_2-l_1)$ units of stock at time $t=0$ for the price $S_0=1$ to get $(\vp^0_0,\vp^1_0)=(1-l_2,l_2)$. 
     This implies that $\cA(1-l_1,l_1)\supseteq \cA(1-l_2,l_2)$ and therefore $v(l_1,w)\geq v(l_2,w)$.

The assertion that $v(l,0)=\log(1-\lambda l)$ follows immediately from $S^0\equiv 1$.
    \vspace{4mm}
    
  (2) As $\tau^{w_1}<\tau^{w_2}$ for all $0\leq w_1<w_2$ and hence $S^{w_1}_t\leq S^{w_2}_t$ for all $t\geq 0$, 
      it is clear that $v(l,w_1)\leq v(l,w_2)$. 
      
    \vspace{4mm}
    
  (3) The continuity of the function $v(\, \cdot\,,w):[0,\frac{1}{\lambda}]\to\R\cup\{-\infty\}$ for fixed $w\geq0$ on $(0,\frac{1}{\lambda})$ follows immediately from the fact 
        that any finitely valued concave function is on the relative interior of its effective domain continuous. 
      At $l=0$, it follows from the fact that $v(\, \cdot\,,w)$ is concave and non-increasing.
      
      The argument for the continuity at $l=\frac{1}{\lambda}$ is slightly more involved. 
      To that end, let $\lambda_n\in(0,1)$ such that $\lambda_n\nearrow\lambda$ and consider for any $n\in\N$ the optimisation problem
      \begin{equation}\label{Pn}
        \E\big[\log\big(V^{\lambda_n,w}_{\tau^w}(\vp^0,\vp^1)\big)\big]\to\max!,\qquad (\vp^0,\vp^1)\in\cA^{\lambda_n,w}(1-l,l),
      \end{equation}
      where $V^{\lambda_n,w}_{\tau^w}(\vp^0,\vp^1):=\vp^0_{\tau^w}+(\vp^1_{\tau^w})^+(1-\lambda_n)S^w_{\tau^w}-(\vp^1_{\tau^w})^-S^w_{\tau^w}$ 
         denotes the terminal liquidation value with transaction costs $\lambda_n$
       and $\cA^{\lambda_n,w}(\vp^0_0,\vp^1_0)$ the set of all self-financing and admissible trading strategies $(\vp^0,\vp^1)$ under transaction costs $\lambda_n$ 
         starting with initial endowment $(\vp^0_{0},\vp^1_{0})$. 
      By Proposition \ref{prop:exist}, the solution $\hvp^n(l,w)=\big(\hvp^{0,n}(l,w),\hvp^{1,n}(l,w)\big)$ to \eqref{Pn} exists 
        for all $(l,w)\in[0,\frac{1}{\lambda_n}]\times[0,\infty)\setminus\{(\frac{1}{\lambda_n},0)\}$ and $n\in\N$. 
      So we can define the functions $v^n:[0,\frac{1}{\lambda_n}]\times [0,\infty)\to\R\cup\{-\infty\}$ for $n\in\N$ by
        $$v^n(l,w):=\sup_{(\vp^0,\vp^1)\in\cA^{\lambda_n,w}(1-l,l)}\E\big[\log\big(V^{\lambda_n,w}_{\tau^w}(\vp^0,\vp^1)\big)\big],$$
        that can by Proposition \ref{prop:exist} be represented as
       \begin{equation}\label{repdualn}
         v^n(l,w)=\inf_{y>0}\inf_{(Z^0,Z^1)\in\mathcal{Z}^{\lambda_n}}\left\{\E\big[-\log(yZ^0_T)-1\big]+y\left(1-l+l\E[Z^1_0]\right)\right\}.
       \end{equation}
      As $\mathcal{Z}^{\lambda_n}\subseteq\mathcal{Z}^{\lambda_{n+1}}$ and $\bigcup_{n=1}^\infty\mathcal{Z}^{\lambda_n}$ is $L^1(\R^2)$-dense in $\mathcal{Z}^{\lambda}$
       and closed under countable convex combinations by martingale convergence, 
       we have by \eqref{repdualn} and Proposition 3.2 in \cite{KS99} that
       \begin{equation} \label{conrepdualn}
         v^{n}(l,w)\searrow v(l,w)
       \end{equation}
       for all $(l,w)\in[0,\frac{1}{\lambda}]\times[0,\infty)$. 
      To see that \eqref{conrepdualn} also holds for $(l,w)=(\frac{1}{\lambda},0)$, choose $(Z^{0,n},Z^{1,n})\equiv(1,1-\lambda_n)\in\mathcal{Z}^{\lambda_n}$. 
      Then $$v^n\left(\tfrac{1}{\lambda},0\right) \leq \inf_{y>0}\left\{-\log(y)-1+y\left(\tfrac{\lambda-\lambda_n}{\lambda}\right)\right\} 
                                                  \leq -\log\big(\tfrac{\lambda}{\lambda-\lambda_n}\big)\to-\infty, $$
          as $n$ goes to infinity. 
      Hence we have for each $w\in[0,\infty)$ a sequence of continuous, non-increasing functions $v^n(\,\cdot\, ,w):[0,\frac{1}{\lambda}]\to \R$ 
       that converges pointwise to the function $v(\,\cdot\, ,w):[0,\frac{1}{\lambda}]\to \R\cup\{-\infty\}$ from above and 
       this already implies that $v(\,\cdot\, ,w)$ is continuous at $\frac{1}{\lambda}$.

      Indeed, let $l_m\in(0,\frac{1}{\lambda})$ such that $l_m\nearrow \frac{1}{\lambda}$ and choose, for $\ve>0$ and $w>0$, some $n\in\N$ such that 
       $0\leq v^n(\frac{1}{\lambda},w)-v(\frac{1}{\lambda},w)\leq\ve$ and then $m(\ve)\in\N$ such that 
       $0\leq v^{n}(l_m,w)-v^{n}(\frac{1}{\lambda},w)\leq\ve$ for all $m\geq m(\ve)$. 
      Since $v^n(l_m,w)\geq v(l_m,w)$, we have that
        $$0\leq v(l_m,w)-v\left(\tfrac{1}{\lambda},w\right)\leq v^n(l_m,w)-v^n\left(\tfrac{1}{\lambda},w\right)+v^n\left(\tfrac{1}{\lambda},w\right)-v\left(\tfrac{1}{\lambda},w\right) \leq 2\ve$$
       for all $m\geq m(\ve)$, which proves the continuity at $l=\frac{1}{\lambda}$ for $w>0$. 
      For $w=0$ and $N\in\N$, choose $n\in\N$ such that $v^n(\frac{1}{\lambda},0)\leq-N$ and then $m(N)\in\N$ such that $0\leq v^{n}(l_m,w)-v^{n}(\frac{1}{\lambda},w)\leq1$ for all $m\geq m(N)$. 
      Using the same arguments as above, we then obatin that $v(l_m,w)\leq -N+1$ for all $m\geq m(N)$, which implies that $\lim_{m\to\infty}v(l_m,0)=-\infty$ and 
       therefore the continuity of $v(\,\cdot\,,0)$ at $l=\frac{1}{\lambda}$.
       
      For the proof of the continuity of $v(l,w)$ in $w$, we observe that $v(l,w)$ is continuous in $l$ for each fixed $w\in[0,\infty)$ and non-decreasing 
       and hence Borel-measurable in $w$ for each fixed $l\in[0,\frac{1}{\lambda}]$. 
      Therefore $v(l,w)$ is a Carath\'eodory function (see Definition 4.50 in \cite{AB}) and hence jointly Borel-measurable by Lemma 4.51 in \cite{AB}. 
      Combining the first part of the proof of Theorem 3.5 in \cite{BT11} with Remark 5.2 in \cite{BT11} this implies that 
       \begin{align} \label{DP1}
         v(l,w) &\leq\sup_{(\vp^0,\vp^1)\in\cA^w(1-l,l)}\E\Big[\log\left(\vp^0_{\tau^w\wedge\sigma}+\vp^1_{\tau^w\wedge\sigma}S^w_{\tau^w\wedge\sigma}\right) \nonumber \\
            & \hspace{35mm} +v\Big(\tfrac{\vp^1_{\tau^w\wedge\sigma}S^w_{\tau^w\wedge\sigma}}{\vp^0_{\tau^w\wedge\sigma}+\vp^1_{\tau^w\wedge\sigma}S^w_{\tau^w\wedge\sigma}},W^w_{\tau^w\wedge\sigma}\Big)\Big]
       \end{align}
       for all stopping times $\sigma$, where we use the joint measurability of $v(l,w)$ to replace the upper-semicontinuous envelope of the value function $V^*$ by the value function $V$ itself 
       (both in the notation of \cite{BT11}).
       
      For $0\leq w_1<w_2$, we then have by \eqref{DP1} that
        \begin{align*}
           0 & \leq v(l,w_2)-v(l,w_1)\\
             & \leq \sup_{(\vp^0,\vp^1)\in\cA^{w_2}(1-l,l)}\E\left[\log\left(\vp^0_{\sigma}+\vp^1_{\sigma}e^\sigma\right)+v\big(L_\sigma(\vp),w_1\big)\right] - v(l,w_1) \\
             & \leq \E\big[\tfrac{\sigma}{\lambda}+v\big(\tfrac{l e^\sigma}{1+l (e^\sigma-1)},w_1\big)\big] -v(l,w_1)
        \end{align*}
       with $\sigma:=\inf\{t>0~|~W^{w_2}_t=w_1\}$, where we used that $L\big(\hvp(l,w_2)\big)\leq \frac{1}{\lambda}$ and $v(l,w)$ is non-increasing in $l$. 
      As $\sigma$ has an inverse Gaussian distribution with mean $\E[\sigma]=(w_2-w_1)$ and variance $\Var[\sigma]=(w_2-w_1)^2$, 
       we can make $v(l,w_2)-v(l,w_1)$ arbitrary small by choosing $w_2$ sufficiently close to $w_1$ using the continuity of $v(\,\cdot\,,w_1)$, 
       which proves the continuity of $v(l,w)$ in $w$ from above.
       
      To prove the continuity of $v(l,w)$ in $w$ from below, consider the stopping time $\rho:=\inf\{t>0~|~W^{w_1}_t=w_2\}$. Then
        \begin{align}
           0 & \leq v(l,w_2)-v(l,w_1) \label{est1} \nonumber \\
             & \leq v(l,w_2)-\E\Big[\Big\{\log\big(1+l( e^\rho-1)\big)+v\left(\tfrac{l e^\rho}{1+l (e^\rho-1)},w_2\right)\Big\}\mathbf{1}_{\{\rho\leq\ve\}} \\
             & \qquad\qquad\qquad\qquad\qquad\qquad\quad      +\big\{\log(\tfrac{1}{\lambda}-1)+\log(\tau\wedge1)\big\}\mathbf{1}_{\{\rho>\ve\}}\Big]  \nonumber
        \end{align}
        for all $\ve>0$ again by \eqref{DP1}, as 
         $$\log\big(V^{liq}_\tau(\vp^0,\vp^1)\big)\geq \log(\tfrac{1}{\lambda}-1)+\log(\tau\wedge 1)$$ 
        for $(\vp^0,\vp^1)\equiv(1-\frac{1}{\lambda},\frac{1}{\lambda})$. 
      Now, since
        \begin{equation*}
          \p[\rho>\ve] \leq \p\left[\sup_{0\leq u\leq \ve}B_u<w_2-w_1+\ve\right] = \p\left[|Z|<\tfrac{w_2-w_1+\ve}{\sqrt{\ve}}\right]
        \end{equation*}
        by the reflection principle for some normally distributed random variable $Z\sim N(0,1)$, 
        we can make the RHS of \eqref{est1} arbitrarily small by choosing $\ve=w_2-w_1$ and $w_1$ sufficiently close to $w_2$ using the continuity of $v(\,\cdot\,,w_2)$.
        
      Having the continuity of $v(l,w)$ in $l$ and $w$ separately, the joint continuity follows from the fact that $v(l,w)$ is non-increasing in $l$ for fixed $w$ and non-decreasing in $w$ for fixed $l$. 
      Indeed, fix $(l,w)\in(0,\frac{1}{\lambda})\times [0,\infty)$ and $\ve>0$ and let $0\leq l_1<l<l_2\leq \frac{1}{\lambda}$ be such that $|v(l',w)-v(l,w)|< \ve$ for all $l'\in[l_1,l_2]$. 
      Now choose $w_1\leq w$ and $w_2>w$ such that $0\leq v(l_2,w)-v(l_2,w_1)<\ve$ and $0\leq v(l_1,w_2)-v(l_1,w)<\ve$. 
      Then 
       \begin{equation*}
         v(l',w')-v(l,w)\leq v(l_1,w_2)-v(l,w)<2\ve
       \end{equation*}
       and 
       \begin{equation*}
        v(l,w)-v(l',w')\leq v(l,w)- v(l_2,w_1)<2\ve
       \end{equation*}       
        for all $(l',w')\in [l_1,l_2]\times [w_1,w_2]$, which gives the joint continuity. 
      If $l=0$, the joint continuity follows by simply choosing $l_1=0$ in the above and, if $l=\frac{1}{\lambda}$ and $w>0$, by setting $l_2=\frac{1}{\lambda}$. 
      To prove the joint continuity for $(l,w)=(\frac{1}{\lambda},0)$, observe that there exists for any $N\in\N$ some $w_1>0$ such that $v(\frac{1}{\lambda},w_1)\leq -N$ and $l_1<\frac{1}{\lambda}$ such 
        that $v(l_1,w_1)-v(\frac{1}{\lambda},w_1)\leq 1$. 
      Then $v(l',w')\leq -N+1$ for all $(l',w')\in[l_1,\frac{1}{\lambda}]\times [0,w_1]$ and hence $v(l,w)$ is also jointly continuous at $(l,w)=(\frac{1}{\lambda},0)$.
      
      \vspace{4mm}
      
  (4) As the value function $v(l,w)$ is jointly continuous, it coincides with its lower-semi\-continuous and upper-semicontinuous envelope. 
      Therefore the dynamic programming principle follows from the weak dynamic programming principle in Theorem 3.5 in \cite{BT11} using Remark 5.2 in \cite{BT11} 
        and observing that the set of controls does not depend on the current time.
      
      \vspace{4mm}
   
  (5) Because $v(l,w)$ is continuous and non-increasing in $l$, the set $\{k\in[0,\frac{1}{\lambda}]~|~v(k,w)=v(0,w)\}$ is a compact interval and so we can define $\ell(w)$ for all $w\geq 0$ via \eqref{def:ell}.
  
      By the joint continuity of $v(l,w)$, we obtain that the function $\ell:[0,\infty)\to [0,\frac{1}{\lambda}]$ is upper semicontinuous and hence c\`adl\`ag, as it is also non-decreasing.
      
      Indeed, suppose by way of contradiction that there exists a sequence $(w_n)$ in $[0,\frac{1}{\lambda}]$ such that $w_n\to w$ and $\lim_{n\to\infty} \ell(w_n)=:k>\ell(w)$ along a subsequence again indexed by $n$. 
      Then $\lim_{n\to\infty}v\big(\ell(w_n),w_n\big)=v(k,w)<v\big(\ell(w),w\big)$ by the joint continuity of $v$ and the definition of $\ell(w)$. 
      But this yields a contradiction, as we also have
        $$\lim_{n\to\infty}v\big(\ell(w_n),w_n\big)=\lim_{n\to\infty}v(0,w_n)=v(0,w)=v\big(\ell(w),w\big)$$
      again using the definition of $\ell(w)$ and the joint continuity of $v$.
      
      To see that $\ell(w)$ is also non-decreasing, denote the optimal strategy to problem \eqref{P1} for $(\vp^0_{0},\vp^1_{0})=(1-l,l)$ and $W_0=w$ by $\hvp(l,w)=\big(\hvp^0(l,w),\hvp^1(l,w)\big)$ 
       and consider $0\leq w_1<w_2$. Then $\hvp\big(\ell(w_2),w_2\big)$ satisfies $L_t\big(\hvp\big(\ell(w_2),w_2\big)\big)\geq \ell(w_1)$ for all $t\leq \sigma:=\inf\{t>0~|~W^{w_2}_t=w_1\}$, 
       as we could otherwise construct a better strategy for the investor trading at $S^{w_2}$ and starting with $(\vp^0_{0},\vp^1_{0})=\big(1-\ell(w_2),\ell(w_2)\big)$. 
      For this, we observe that
        $$d L_t(\vp)=L_t(\vp)\big(1-L_t(\vp)\big)\mathbf{1}_{\llbracket0,\tau\rrbracket}dt+\frac{L_t(\vp)}{\vp^1_t}d\vp^{1,\uparrow}_t-\frac{L_t(\vp)\big(1-\lambda L_t(\vp)\big)}{\vp^1_t}d\vp^{1,\downarrow}_t,$$
        which implies that we can always trade in such a way to keep the leverage $L_t(\vp)\equiv \ell(w_1)$. 
      For $\ell(w_1)>1$, we buy stocks at the rate $d\vp^{1,\uparrow}_t=-\vp^1_t\big(1-L_t(\vp)\big)\mathbf{1}_{\llbracket0,\tau\rrbracket}dt$ and 
        for $\ell(w_1)<1$ we sell at $-d\vp^{1,\downarrow}_t=-\vp^1_t\frac{(1-L_t(\vp))}{(1-\lambda L_t(\vp))}\mathbf{1}_{\llbracket0,\tau\rrbracket}dt$. 
      This gives $d\log(\vp^0_t+\vp^1_tS_t)=\ell(w_1)\mathbf{1}_{\llbracket0,\tau\rrbracket}dt$ and 
                    $d\log(\vp^0_t+\vp^1_tS_t)=\ell(w_1)\frac{1-\lambda}{1-\lambda\ell(w_1)}\mathbf{1}_{\llbracket0,\tau\rrbracket}dt$, respectively. 
      As $\frac{1-\lambda}{1-\lambda\ell(w_1)}>1$ for $\ell(w_1)<1$, we obtain by part (4) that 
        the strategy $\vp=(\vp^0,\vp^1)\in\cA^{w_2}\big(1-\ell(w_2),\ell(w_2)\big)$ 
          that keeps $L_t(\vp)= L_t\big(\hvp\big(\ell(w_2),w_2\big)\big)\vee \ell(w_1)$ for all $t\leq \sigma$ and 
          then continues with $\hvp\big(\ell(w_1),w_1\big)$, if $L_t\big(\hvp\big(\ell(w_2),w_2\big)\big)\leq \ell(w_1)$, or $\hvp\big(\ell(w_2),w_2\big)$, if $L_t\big(\hvp\big(\ell(w_2),w_2\big)\big)> \ell(w_1)$, 
        yields a higher expected utility, i.e.,
           $$\E\big[\log\big(V^{liq}_{\tau^{w_2}}(\vp^0,\vp^1)\big)\big] > \E\Big[\log\Big(V^{liq}_{\tau^{w_2}}\big(\hvp^0\big(\ell(w_2),w_2\big),\hvp^1\big(\ell(w_2),w_2\big)\big)\Big)\Big].$$
           
      As $v(l,w)=v\big(\ell(w),w\big)$ for $l\in[0,\ell(w)]$ and $v(l,w)<v\big(\ell(w),w\big)$ for $l\in\big(\ell(w),\frac{1}{\lambda}\big]$, 
        it follows from the concavity of $v(l,w)$ in $l$ that $v(l,w)$ is strictly decreasing in $l$ on $\big(\ell(w),\frac{1}{\lambda}\big]$. 
      This implies that 
         $$g(l_1,w):=V^{liq}_\tau\big(\hvp^0(l_1,w),\hvp^1(l_1,w)\big)\ne V^{liq}_\tau\big(\hvp^0(l_2,w),\hvp^1(l_2,w)\big)=:g(l_2,w)$$
        for $\ell(w)<l_1<l_2\leq\frac{1}{\lambda}$ and hence the strict concavity of $v(l,w)$ in $l$ on $\big(\ell(w),\frac{1}{\lambda}\big]$, as
        \begin{align*}
           \mu v(l_1,w)+ (1-\mu) v(l_2,w)&=\mu \E\big[\log\big(g(l_1,w)\big)\big]+ (1-\mu) \E\big[\log\big(g(l_2,w)\big)\big]\\
                                         &< \E\big[\log\big(\mu g(l_1,w)+(1-\mu)g(l_2,w)\big)\big]  \\
                                         &\leq v\big(\mu l_1+(1-\mu)l_2,w\big)
        \end{align*}
        for all $\mu\in(0,1)$ by Jensen's inequality.
 \end{proof}

 \begin{lemma}\label{l:ell}
  Let $\ell : [0,\infty)\to[0,\frac{1}{\lambda}]$ be an increasing function (no left- or right-continuity is assumed). 
  Recall that the optimizer $\hvp=(\hvp^0_t,\hvp^1_t)_{t\geq 0}$ is right-continuous and that we have to distinguish between $\hvp_{0-}$ and $\hvp_{0}$.
  
  If
   \begin{equation}\label{D1}
      \p\left[\inf_{0\leq t< \tau}\big(L_t(\hvp)-\ell(W_t)\big)<0\right]>0,
   \end{equation}
  then there are stopping times $0\leq\sigma_1\leq\sigma_2$ and $\alpha>0$, such that $\p[\sigma_1<\sigma_2\leq \tau]>0$ and $L_t(\hvp)<\ell(W_t)-\alpha$ on $\rrbracket\sigma_1,\sigma_2\rrbracket$.
 \end{lemma}
 
 \begin{proof}
  Assuming \eqref{D1}, there is $\ve>0$ such that $\sigma:=\inf\{t>0~|~L_t(\hvp)<\ell(W_t)-\ve\}$ satisfies $\p[\sigma<\tau]>0$. 
  To see that $\sigma$ is a stopping time, we observe that it is the first hitting time of the progressively measurable set $\big\{(\omega,t)~\big|~L_t(\hvp)(\omega)<\ell\big(W_t(\omega)\big)-\ve\big\}$. 
  By the c\`adl\`ag property of $\hvp$ we have $L_\sigma(\hvp)\leq \lim_{w\searrow W_\sigma}\ell(w)-\ve$ on $\{\sigma <\tau\}$. 
  Now we distinguish two cases.

  Case 1: Let $A:=\{\omega~|~\text{$\ell$ has a continuity point at $W_\sigma$}\}$ and
    \begin{equation}\label{D2}
       \p[ A,~\sigma<\tau ]>0.
    \end{equation}
    Define $\sigma_1:=\sigma\mathbf{1}_A+\infty\mathbf{1}_{A^c}$ and the Borel-measurable function $\delta(w)$ by
      $$ \delta(w):=\frac{\sup\big\{|w'-w|~\big|~\ell(w')\geq \ell(w)-\tfrac{\ve}{3}\big\}}{2} $$
      so that $\delta(W_\sigma)>0$ on $A\cap\{\sigma<\tau\}$ and $\ell(w')>\ell(w)-\frac{\ve}{3}$ for every $w'\geq w-\delta(w)$. 
    As regards the process $L_t(\hvp)$ let
      $$ \varrho:=\inf\big\{t>\sigma~\big|~L_t(\hvp)>L_\sigma(\hvp)+\tfrac{\ve}{3}\big\}. $$
    We cannot deduce that $L_{\varrho}(\hvp)\leq L_\sigma(\hvp)+\frac{\ve}{3}$, as $L_t(\hvp)$ may have an upwards jump at time $\varrho$. 
    To remedy this difficulty, we may use the fact that the stopping time $\varrho$ is {\it predictable}, as every stopping time in a Brownian filtration is predictable (see e.g.~\cite[Example 4.12]{N06}). 
    We therefore may find an increasing sequence $(\varrho_n)_{n=1}^\infty$ of announcing stopping times, i.e., $\varrho_n<\varrho$ and $\lim_{n\to\infty}\varrho_n=\varrho$, almost surely. 
    As $\varrho>\sigma_1$ on $A$ we may find $n$ such that $\p[\{\varrho_n>\sigma_1\}\cap A]>0$. 
    For this $n$, we may define 
      $$\sigma_2:=\inf\{t>\sigma_1~|~W_t\leq W_\sigma-\delta(W_\sigma)\}\wedge\varrho_n\wedge\tau$$
      on $A\cap\{\varrho_n>\sigma_1\}$ and $+\infty$ elsewhere. 
    Then $\sigma_1<\sigma_2$ on $A$ and $\sigma_1$, $\sigma_2$ and $\alpha=\tfrac{\ve}{3}$ satisfy the assertion of the lemma.
    
  Case 2: If \eqref{D2} fails, there must be one point $\widetilde{w}\in(0,\infty)$ with $\lim_{w \nearrow \widetilde{w}}\ell(w)<\lim_{w \searrow \widetilde{w}}\ell(w)$ such that $\p[W_\sigma= \widetilde{w}]>0$. 
          For each real number $w>\widetilde{w}$, we define the stopping time $\sigma^w$ by
                $$\sigma^w:=\inf\{t>\sigma~|~W_t=w\}.$$
          We may find $w>\widetilde{w}$ which is a continuity point of $\ell$ and sufficiently close to $\widetilde{w}$ such that $\p[\sigma^w<\tau]>0$. 
          We may then proceed as in Case 1 by letting $\sigma_1:=\sigma^w$, which completes the proof.
 \end{proof}

 \begin{proposition}\label{prop:opt}
   The optimal strategy $\hvp=(\hvp^0_t,\hvp^1_t)_{t\geq 0}$ is determined by the non-decreasing function $\ell:[0,\infty)\to[0,\frac{1}{\lambda}]$ in \eqref{def:ell} in the following way:
    \begin{enumerate}[(i)]
     \item $(\hvp^1_t)_{0\leq t<\tau}$ is non-decreasing while $(\hvp^0_t)_{0\leq t<\tau}$ is non-increasing and satisfies
                  $$d\hvp^0_t=-S_td\hvp^1_t=-e^td\hvp^1_t,\qquad 0\leq t<\tau.$$
     \item $(\hvp^1_t)_{0\leq t<\tau}$ is the smallest non-decreasing process such that 
            \begin{equation} \label{eq:prop:opt}
              L_t(\hvp)=\frac{\hvp^1_t e^t}{1+\int_0^t\hvp^1_u e^udu} \geq \ell(W_t),\qquad 0\leq t<\tau.
            \end{equation}
    \end{enumerate}
 \end{proposition}
 
 \begin{proof}
   (i) This follows immediately from the following fact: 
       As $S$ is strictly increasing on $\llbracket0,\tau\rrbracket$, 
         any strategy selling stock shares before time $\tau$ sells them at a lower price and hence has a smaller liquidation value at time $\tau$ as the strategy not selling stock shares before time $\tau$.
         
       Here is the formal argument. 
       Let $(\vp^0,\vp^1)\in\cA(\vp^0_0,\vp^1_0)$ and $\vp^1=\vp^1_0+\vp^{1,\uparrow}-\vp^{1,\downarrow}$ the Jordan--Hahn decomposition of $\vp^1$ into two non-decreasing processes 
         $\vp^{1,\uparrow}$ and $\vp^{1,\downarrow}$ starting at $0$. 
       Define a strategy $(\tvp^0,\tvp^1)\in\cA(\vp^0_0,\vp^1_0)$ by 
        \begin{equation*}
          \tvp^1=\vp^1_0+\vp^{1,\uparrow} \quad\mbox{ and }\quad \tvp^0=\vp^0_0-\int S_ud\vp^{1,\uparrow}_u.
        \end{equation*}
       Then 
        \begin{align*}
          V^{liq}_\tau(\vp^0,\vp^1) &=\vp^0_0+\int_0^\tau(1-\lambda) S_ud\vp^{1,\downarrow}_u-\int_0^\tau S_ud\vp^{1,\uparrow}_u+(\vp^1_\tau)^+(1-\lambda)S_\tau-(\vp^1_\tau)^-S_\tau\\
                                    &\leq\vp^0_0-\int_0^\tau S_ud\vp^{1,\uparrow}_u+(\vp^1_0+\vp^{1,\uparrow}_\tau)^+(1-\lambda)S_\tau-(\vp^1_0+\vp^{1,\uparrow}_\tau)^-S_\tau \\
                                    &=V^{liq}_\tau(\tvp^0,\tvp^1),
        \end{align*}
        since $\vp^1_\tau\leq \tvp^1_\tau=\vp^1_\tau+\vp^{1,\downarrow}_\tau$ and $S$ is non-decreasing and therefore
            $$\int_0^\tau(1-\lambda) S_ud\vp^{1,\downarrow}_u+(\vp^1_\tau)^+(1-\lambda)S_\tau-(\vp^1_\tau)^-S_\tau\leq(\tvp^1_\tau)^+(1-\lambda)S_\tau$$
        for $\tvp^1_\tau\geq0$ and 
            $$\int_0^\tau(1-\lambda) S_ud\vp^{1,\downarrow}_u-(\vp^1_\tau)^-S_\tau\leq-(\tvp^1_\tau)^-S_\tau$$ 
        for $\tvp^1_\tau<0$.
        
    \vspace{4mm}
        
  (ii) That $(\hvp^1_t)_{0\leq t<\tau}$ is a non-decreasing process such that $L_t(\hvp)\geq \ell(W_t)$ for $0\leq t<\tau$ follows immediately from part (i) above and 
           by combining Lemmas \ref{l:propv} and \ref{l:ell}. 
      Indeed, suppose that 
           $$ \p\left[\inf_{0\leq t< \tau}\big(L_t(\hvp)-\ell(W_t)\big)<0\right]>0.$$ 
      Then there exist two stopping times $\sigma_1$ and $\sigma_2$ and $\alpha>0$ such that $\p[\sigma_1<\sigma_2\leq \tau]>0$ 
           and $L_t(\hvp)<\ell(W_t)-\alpha$ on $\rrbracket\sigma_1,\sigma_2\rrbracket$ by Lemma \ref{l:ell}. 
      Therefore, we can define a strategy $\tvp$ such that $\tvp=\hvp$ on $\llbracket 0, \sigma_1\rrbracket$ and $L_t(\tvp)= L_t(\hvp)+\alpha$ on $\llbracket  \sigma_1, \sigma_2\rrbracket$. 
      Then 
           \begin{equation*}
              \begin{aligned}
                 \E\big[\log & (\tvp^0_{\sigma_2}+\tvp^1_{\sigma_2}S_{\sigma_2})+v\big(L_{\sigma_2}(\tvp),W_{\sigma_2}\big)\big] \\
                             & = \E\left[\int_0^{\sigma_2} L_t(\hvp)dt +\alpha(\sigma_2-\sigma_1)+v\big(L_{\sigma_2}(\hvp),W_{\sigma_2}\big)\right]\\
                             & = v(l, w)+\alpha E[\sigma_2-\sigma_1]>v(l,w)
              \end{aligned}
           \end{equation*}
        by part (4) of Lemma \ref{l:propv}, since $L_{\sigma_2}(\hvp)\leq L_{\sigma_2}(\tvp)\leq \ell(W_{\sigma_2})$ and $v(\,\cdot\,,W_{\sigma_2})$ is constant on $[0,\ell(W_{\sigma_2})]$. 
      But this contradicts the optimality of $\hvp$ by part (4) of Lemma \ref{l:propv}.
        
      To see that $\Delta L(\hvp)= \ell (W)-L_{-}(\hvp)$, assume by way of contradiction that there exists a stopping time $\sigma$ 
          such that $P(A)>0$ for $A:=\{\Delta L_{\sigma\wedge\tau}(\hvp)>\ell (W_{\sigma\wedge\tau})-L_{\sigma\wedge\tau-}(\hvp)\geq0\}$. 
      Then we have
         \begin{align*}
           v\big(L_{\sigma\wedge\tau}(\hvp),W_{\sigma\wedge\tau}\big)&=v\big(L_{\sigma\wedge\tau-}(\hvp)+\Delta L_{\sigma\wedge\tau}(\hvp),W_{\sigma\wedge\tau}\big)\\
           &<v\big(L_{\sigma\wedge\tau-}(\hvp)+\big(\ell (W_{\sigma\wedge\tau})-L_{\sigma\wedge\tau-}(\hvp)\big),W_{\sigma\wedge\tau}\big) \\
           &=v\big(\ell (W_{\sigma\wedge\tau}),W_{\sigma\wedge\tau}\big)
         \end{align*}
        on $A$, as $v(l,w)$ is strictly decreasing on $(\ell(w),\frac{1}{\lambda}]$. 
      But this contradicts the optimality of $\hvp$ by part 4) of Lemma \ref{l:propv}. 
      Indeed, the strategy $(\tvp^0,\tvp^1)\in\cA(\vp^0_0,\vp^1_0)$ given by
         $$d\tvp^1_t=\mathbf{1}_{\llbracket0,\sigma\wedge\tau\llbracket}d\hvp^1_t+\mathbf{1}_{\llbracket\sigma\wedge\tau\rrbracket}
                 \left(\frac{\ell(W_{\sigma\wedge\tau})(\hvp^0_{\sigma\wedge\tau-}+\hvp^1_{\sigma\wedge\tau-}e^{\sigma\wedge\tau})}{e^{\sigma\wedge\tau}}-\hvp^1_{\sigma\wedge\tau-}\right)$$
        and $d\tvp^0=-Sd\tvp^1$ 
        satisfies $L_t(\tvp)=L_t(\hvp)$ on $\llbracket0,\sigma\wedge\tau\llbracket$ and $L_{\sigma\wedge\tau}(\tvp)=\ell (W_{\sigma\wedge\tau})$ and 
        therefore yields
         \begin{align*}
           v(l,w) &= \E\left[\log\left(\hvp^0_{\sigma\wedge\tau}+\hvp^1_{\sigma\wedge\tau}S_{\sigma\wedge\tau}\right)+v\left(L_{\sigma\wedge\tau}(\hvp),W_{\sigma\wedge\tau}\right)\right] \\
                  &< \E\left[\log\left(\tvp^0_{\sigma\wedge\tau}+\tvp^1_{\sigma\wedge\tau}S_{\sigma\wedge\tau}\right)+v\left(L_{\sigma\wedge\tau}(\tvp),W_{\sigma\wedge\tau}\right)\right],
         \end{align*}
        where we used that $\tvp^0_{\sigma\wedge\tau}+\tvp^1_{\sigma\wedge\tau}S_{\sigma\wedge\tau}=\hvp^0_{\sigma\wedge\tau}+\hvp^1_{\sigma\wedge\tau}S_{\sigma\wedge\tau}$. 
       Since $L_t(\hvp)\geq \ell(W_t)$ for all $0\leq t<\tau$, this proves $\Delta L(\hvp)= \ell (W)-L_{-}(\hvp)$.
        
       Let $\hvp\in\cA(\vp^0_0,\vp^1_0)$ be the solution and $\tvp\in\cA(\vp^0_0,\vp^1_0)$ be the strategy such that $(\tvp_t^1)_{0\leq t<\tau}$ is the smallest non-decreasing process 
          with $L_t(\tvp)\geq \ell(W_t)$ for all $0\leq t<\tau$. 
       Define a non-negative predictable process $(\tpsi_t)_{0\leq t<\tau}$ of finite variation by $\tpsi_t:= L_t(\hvp)-L_t(\tvp)$ and suppose by way of contradiction that
          \begin{equation} \label{eq:cont1}
             \p\left[\sup_{0\leq t<\tau} \tpsi_t>\ve\right]>0
          \end{equation}
         for some $\ve>0$ or, equivalently, that $\p\left[\tau_{\ve}<\tau\right]>0$ for the stopping time
          $$\tau_\ve:=\inf\big\{t>0~|~\tpsi_t>\ve\big\}\wedge\tau.$$
       Next observe that
          \begin{equation*}
            \Delta L_t(\tvp)\geq \ell(W_t)-L_{t-}(\hvp)=\Delta L_t(\hvp)
          \end{equation*}
         for all $0\leq t<\tau$, since $L_t(\hvp)\geq L_t(\tvp)\geq \ell(W_t)$ for all $0\leq t<\tau$ and $L(\tvp)$ and $L(\hvp)$ also only jump upwards. 
       This implies that $\tpsi^{\uparrow}$ is continuous, where $\tpsi=\tpsi^{\uparrow}-\tpsi^{\downarrow}$ denotes the Jordan-Hahn decomposition of $\tpsi$, and 
         therefore that $L_{\tau_{\ve}}(\hvp)=L_{\tau_{\ve}}(\tvp)+\ve$.
       
       Now consider the trading strategy $\vp\in\cA(\vp^0_0,\vp^1_0)$ such that $\vp^1=\hvp^1$ on $\llbracket 0,\tau_\ve\rrbracket$ and buys the minimal amount to keep 
         $L_t(\vp)\geq \ell(W_t)$ on $\rrbracket \tau_{\ve},\tau\rrbracket$ and $d\vp^0=Sd\vp^1$. 
       Define, similarly as above, a non-negative predictable process $(\psi_t)_{0\leq t<\tau}$ of finite variation by $\psi_t:= L_t(\hvp)-L_t(\vp)$ and the stopping times 
         $$\tau_{\ve,h}:=\inf\{t>\tau_{\ve}~|~\psi_t>h\}\wedge\tau, \qquad h>0,$$
        that satisfy $L_{\tau_{\ve,h}}(\hvp)=L_{\tau_{\ve,h}}(\vp)+h$ on $\{\tau_{\ve,h}<\tau\}$ and $\tau_{\ve,h}\searrow\tau_\ve$ for $h\searrow 0$ on $\{\tau_\ve<\tau\}$, 
        since $\psi^{\uparrow}$ is again continuous. 
       Then we have by the optimality of $\hvp$ and by the part (4) of Lemma \ref{l:propv} that
         \begin{equation} \label{eq:cont2}
            \frac{\E\left[\int_{\tau_\ve}^{\tau_{\ve,h}}\big(L_s(\hvp)-L_s(\vp)\big) ds+
                  v\big(L_{\tau_{\ve,h}}(\hvp),W_{\tau_{\ve,h}}\big)-v\big(L_{\tau_{\ve,h}}(\vp),W_{\tau_{\ve,h}}\big)\Big|\F_{\tau_\ve}\right]}{h}\geq 0
         \end{equation}
        on $\{\tau_\ve<\tau\}$ for all $h>0$. On the other side, we have 
          $$\lim_{h\searrow0}\frac{\E\left[\int_{\tau_\ve}^{\tau_{\ve,h}}\big(L_s(\hvp)-L_s(\vp)\big) ds\Big|\F_{\tau_\ve}\right]}{h}
                   \leq\lim_{h\searrow0}\E\left[(\tau_{\ve,h}-\tau_\ve)|\F_{\tau_\ve}\right]=0$$
        on $\{\tau_\ve<\tau\}$ by Lebesgue's dominated convergence theorem and 
          $$ \frac{\E\left[v\big(L_{\tau_{\ve,h}}(\hvp),W_{\tau_{\ve,h}}\big)-v\big(L_{\tau_{\ve,h}}(\vp),W_{\tau_{\ve,h}}\big)\big|\F_{\tau_\ve}\right]}{h}
                                                             \leq \E\left[v'_-\big(L_{\tau_{\ve,h}}(\vp),W_{\tau_{\ve,h}}\big)\Big|\F_{\tau_\ve}\right] $$
        on $\{\tau_\ve<\tau\}$, since $L_{\tau_{\ve,h}}(\hvp)-L_{\tau_{\ve,h}}(\vp)= h$ on $\{\tau_{\ve,h}<\tau\}$. 
       As $$v'_-(l,w):=\inf_{h>0}\frac{v(l,w)-v(l-h,w)}{h}$$ is as the infimum of continuous functions upper-semicontinuous and
           $$L_{\tau_{\ve}}(\vp)= L_{\tau_{\ve}}(\tvp)+\ve\geq \ell(W_{\tau_{\ve}})+\ve$$
        on $\{\tau_\ve<\tau\}$, we obtain
           $$ \lim_{h\searrow0}\E\left[v'_-\big(L_{\tau_{\ve,h}}(\vp),W_{\tau_{\ve,h}}\big)\Big|\F_{\tau_\ve}\right]
                    \leq v'_-\big(L_{\tau_\ve}(\vp),W_{\tau_\ve}\big)\leq v'_-\big(\ell(W_{\tau_{\ve}})+\ve,W_{\tau_\ve}\big)<0 $$
        on $\{\tau_\ve<\tau\}$ by Fatou's Lemma, which is a contradiction to \eqref{eq:cont2} and hence \eqref{eq:cont1}.
 \end{proof}

 The following result is the crucial property of the function $\ell$.
 \begin{lemma}\label{cl}
  There is $\overline{w}$ such that $\ell(w)=\frac{1}{\lambda}$ for all $w\geq \overline{w}$.
 \end{lemma}

 \begin{proof}
  Suppose to the contrary that $\ell(w)<\frac{1}{\lambda}$ for all $w\geq0$. It is straightforward to check that $\lim_{w\to\infty}\ell(w)=\frac{1}{\lambda}$.
  
  The basic idea is now to construct a strategy $\overline\vp$ that yields, for sufficiently large $W_0=w$, a higher expected utility than the optimal strategy $\hvp$ and hence a contradiction proving the lemma.
  
  For this, we define the strategy $\overline\vp$ in the following way: 
  We start with $(\overline\vp^0_0,\overline\vp^1_0)= (1-\frac{1}{\lambda},\frac{1}{\lambda})$, i.e., with maximal leverage $L_0(\overline\vp)=\frac{1}{\lambda}$, 
    continue to leave $(\overline\vp^0_t,\overline\vp^1_t)$ constant until the stopping time $$\vr:=\inf\{t>0~|~L_t(\overline\vp)=L_t(\hvp)\}$$
    and trade such that $L_t(\overline\vp)=L_t(\hvp)$ after time $\vr$. 
  Note that the strategy $\hvp$ only trades at time $t<\tau$, if $L_t(\hvp)=\ell(W_t)$, by part (ii) of Proposition \ref{prop:opt} and $L_{t_1}(\overline\vp)>L_{t_1}(\hvp)$, 
    if $L_{t_0}(\overline\vp)>L_{t_0}(\hvp)$ and $\hvp$ does not trade between $t_1$ and $t_0$ for $0\leq t_0\leq t_1<\tau$, 
    which follows by a direct computation. 
  Combing both we obtain that $L_t(\overline\vp)>L_t(\hvp)\geq \ell(W_t)$ for $0\leq t<\vr$ and $L_\vr(\overline\vp)=L_\vr(\hvp)=\ell(W_\vr)$. 
  Using the decreasing function 
      $$f(t):=\tfrac{\frac{1}{\lambda}e^t}{1-\frac{1}{\lambda}+\frac{1}{\lambda}e^t}$$ 
    starting at $f(0)=\frac{1}{\lambda}$ and satisfying $f(t)=L_t(\overline\vp)$ for $0\leq t\leq \vr$ 
    and the ``obstacle function'' $$b(t):=\ell^{-1}\big(f(t)\big)$$
    then allows us to rephrase the definition of $\vr$ as $\vr=\inf\{t>0~|~W_t=b(t)\}$. 
  Here $\ell^{-1}(\cdot)$ denotes the right-continuous generalised inverse.
  
  As $b:(0,\infty)\to(0,\infty)$ is non-increasing and satisfies $\lim_{t\searrow 0} b(t)=\infty$, 
   we obtain a sequence $(a_n)_{n=1}^{\infty}$ of non-positive numbers with $\sum_{n=1}^\infty a_n=\infty$ by setting $a_n:=b(2^{-n})-b(2^{-n+1})$. 
  Hence we may find, for any $\ve>0$, a number $n$ such that 
   \begin{equation} \label{star}
     \ve a_n > 2^{-n/4}, 
   \end{equation}
   as $\ve\sum_{n=1}^\infty a_n=\sum_{n=1}^\infty \ve a_n\leq \sum_{n=1}^\infty 2^{-n/4} <\infty$ would lead to a contradiction otherwise. 
  Now we estimate
      $$ \p[\vr>2^{-n+1}~|~\vr>2^{-n}] $$
   with $W_0=w_n=\frac{a_n}{2}+b(2^{-n+1})$ which becomes small, if $\frac{2^{-n/2}}{a_n}$ becomes small. 
  By \eqref{star} we have 
      $$\tfrac{2^{-n/2}}{a_n} < \ve  2^{-n/4},$$
   so that by elementary estimates on the Gaussian distribution, we have that
   \begin{equation}\label{B4}
     \p[\vr>2^{-n+1}~|~\vr>2^{-n}]<\delta 2^{-2n},
   \end{equation}
   for a pre-given $\delta>0$. 
  To see this, observe that 
   \begin{equation} \label{star2}
     \p[\vr>2^{-n+1}~|~\vr>2^{-n}]=\frac{\p[\vr>2^{-n+1}]}{\p[\vr>2^{-n}]}\leq \frac{\p[W_{2^{-n+1}}\leq b(2^{-n+1})]}{\p\left[\sup\limits_{0\leq u \leq 2^{-n}}W_u< b(2^{-n})\right]},
   \end{equation}
   where we can estimate the probabilities on the right-hand side separately.
  
  As 
    \begin{equation*}
       \p\Bigg[\sup_{0\leq u \leq 2^{-n}}W_u < b(2^{-n})\Bigg]\geq \p\Bigg[\sup_{0\leq u \leq 2^{-n}}B_u < b(2^{-n})-w_n\Bigg],
    \end{equation*}
  we obtain by the reflection principle that
  \begin{align*}
     \p\left[\sup\limits_{0\leq u\leq 2^{-n}}W_u < b(2^{-n})\right] & \geq 1- \p\left[\sup_{0\leq u \leq 2^{-n}}B_u \geq b(2^{-n})-w_n\right] \\
             &= 1- 2\p\left[B_{2^{-n}} \geq \frac{a_n}{2} \right] = 1-\p\left[|Z| \geq \frac{1}{2}\frac{a_n}{2^{-n/2}}\right]
  \end{align*}
   for a standard normal distributed random variable $Z\sim N(0,1)$ and therefore
   \begin{equation}\label{pr:ml:est1}
      \p\left[\sup\limits_{0\leq u \leq 2^{-n}}W_u < b(2^{-n})\right]\geq 1-\left(2 \frac{2^{-n/2}}{a_n}\right)^2 > 1-\left(2 \ve 2^{-n/4}\right)^2
   \end{equation}
   by applying Chebyscheff's inequality with $\E[Z^2]=1$. 
   
  For $\ve>0$ sufficiently small such that $a_n^3\ve^4\leq \frac{1}{8}$, we have $$-\frac{a_n}{2}+2(\ve a_n)^4\leq -\frac{a_n}{4}.$$ 
  Hence for the second probability we obtain
   \begin{align*}
      \p\left[W_{2^{-n+1}}\leq b(2^{-n+1})\right] &= \p\left[B_{2^{-n+1}}\leq b(2^{-n+1})-w_n + 2^{-n+1}\right] \\
          &\leq \p\left[B_{2^{-n+1}}\leq b(2^{-n+1})-w_n + 2(\ve a_n)^4\right]\\
          &=\p\left[\sqrt{2^{-n+1}}Z\leq -\frac{a_n}{2}+ 2(\ve a_n)^4\right] \\
          &\leq \p\left[Z\leq -\frac{a_n}{4\sqrt{2}2^{-n/2}}\right] = \frac{1}{2}\p\left[|Z|\geq\frac{a_n}{4\sqrt{2}2^{-n/2}}\right]
   \end{align*}
   with a standard normal distributed random variable $Z\sim N(0,1)$.
  Then, applying again Chebyscheff's inequality this time with $\E[Z^8]=105$ gives
   \begin{equation} \label{pr:ml:est2}
    \begin{aligned} 
      \p\left[W_{2^{-n+1}}\leq b(2^{-n+1})\right] &\leq \frac{1}{2}\cdot 105\big(4\sqrt{2}\big)^8\Big(\frac{2^{-n/2}}{a_n}\Big)^8   \\
              &\leq \frac{1}{2}\cdot \left(105\big(4\sqrt{2}\big)^8\ve^8\right)2^{-2n} =:\frac{1}{2}\delta 2^{-2n}.
    \end{aligned}
   \end{equation}
  Plugging \eqref{pr:ml:est1} and \eqref{pr:ml:est2} into \eqref{star2} then yields \eqref{B4} after choosing $\ve$ small enough such that 
     $$ \p\left[\sup\limits_{0\leq u \leq 2^{-n}}W_u < b(2^{-n})\right]\geq \frac{1}{2}.$$
     
  On the set $\{\vr<\infty\}$ we can estimate the positive effect of the strategy $\overline\vp$ on the value function by
  \begin{equation*}
    \begin{aligned}
        \E\Big[\Big(&\log\big(V^{liq}_{\tau}(\overline{\vp}^0,\overline{\vp}^1)\big)-\log\big(V^{liq}_{\tau}(\hvp^0,\hvp^1)\big)\Big)
                  \mathbf{1}_{\{\vr<\infty\}}\Big] \\
            & \geq \E\left[\int_0^\vr\big(L_s(\overline\vp)-L_s(\hvp)\big)ds\mathbf{1}_{\{\vr<\infty\}}\right]\\
            & \geq \E\left[\int_0^{2^{-n}}\big(L_s(\overline\vp)-L_s(\hvp)\big)ds\mathbf{1}_{\{2^{-n}<\vr\leq 2^{-n+1} \}}\right].
    \end{aligned}
  \end{equation*}   
   Using that 
      $$\max_{0\leq u\leq 2^{-n}}L_u(\hvp) = \max_{0\leq u\leq 2^{-n}}\ell(W_u)\leq \ell\big(b(2^{-n})\big)=f(2^{-n})=L_{2^{-n}}(\overline\vp)$$
    on $\big\{\sup_{0\leq u \leq 2^{-n}}W_u < b(2^{-n})\big\}$ and that 
      $$\p[2^{-n}<\vr\leq 2^{-n+1} ]=\p[\vr>2^{-n}]\cdot(1- \p[\vr>2^{-n+1}~|~\vr>2^{-n}])\geq\frac{1}{2} \p[\vr>2^{-n}]$$
    and 
      $$ \p\left[\sup\limits_{0\leq u \leq 2^{-n}}W_u < b(2^{-n})\,\Bigg|\, 2^{-n}<\vr\leq 2^{-n+1}\right]\geq\frac{1}{2}$$ 
    by \eqref{pr:ml:est1} for sufficiently large $n$, we get 
    \begin{equation*}
      \begin{aligned}
        \E\Big[\Big(&\log\big(V^{liq}_{\tau}(\overline{\vp}^0,\overline{\vp}^1)\big)-\log\big(V^{liq}_{\tau}(\hvp^0,\hvp^1)\big)\Big)
                  \mathbf{1}_{\{\vr<\infty\}}\Big]\\
            & \geq \int_0^{2^{-n}}\big(f(s)-f(2^{-n})\big)ds \cdot \frac{1}{4} \p[\vr>2^{-n}].
      \end{aligned}     
    \end{equation*}

   As $$f(s)-f(2^{-n})\geq \min_{u\in[0,2^{-n}]}\big(-f'(u)\big)(2^{-n}-s)$$ for $s\in[0,2^{-n}]$ 
    and $f'(u)=f(u)\big(1-f(u)\big)$ satisfies 
      $$-f'(u)\geq \tfrac{1}{2}f(0)\big(f(0)-1\big)= \tfrac{1}{2}\tfrac{1}{\lambda}\big(\tfrac{1}{\lambda}-1\big)$$ 
      for all $u\in[0,2^{-n}]$ by continuity of $f$ for sufficiently large $n$, 
    we obtain that
    \begin{align} \label{pr:ml:est3}
      &\E\left[\Big(\log\big(V^{liq}_{\tau}(\overline{\vp}^0,\overline{\vp}^1)\big)-\log\big(V^{liq}_{\tau}(\hvp^0,\hvp^1)\big)\Big)
                  \mathbf{1}_{\{\vr<\infty\}}\right]\geq c_1 2^{-2n}\p[\vr>2^{-n}]
    \end{align}
    for sufficiently large $n$ with $c_1:=\frac{1}{16} \frac{1}{\lambda}\left(\frac{1}{\lambda}-1\right)>0$.

   For the estimate of the negative effect of the strategy $\overline\vp$ on the set $\{\vr=\infty\}$, we observe that, if $V^{liq}_\tau(\hvp^0,\hvp^1)\geq1$, then
      $$1\leq V^{liq}_\tau(\hvp^0,\hvp^1)\leq 1+\hvp^1_\tau\big((1-\lambda)S_\tau-1\big)\leq 1+\overline\vp^1_0\big((1-\lambda)S_\tau-1\big)= V^{liq}_{\tau}(\overline\vp^0,\overline\vp^1),$$
     since $\hvp^1_t\leq \overline\vp^1_0=\frac{1}{\lambda}$ for all $0\leq t<\vr$, and therefore 
     \begin{equation*}
       \log\big(V^{liq}_{\tau}(\overline{\vp}^0,\overline{\vp}^1)\big)-\log\big(V^{liq}_{\tau}(\hvp^0,\hvp^1)\big)\geq 0
     \end{equation*}
       on $\{\vr=\infty,~V^{liq}_\tau(\hvp^0,\hvp^1)\geq 1\}$. 
   Hence it is sufficient to consider $\{\vr=\infty,~V^{liq}_\tau(\hvp^0,\hvp^1)< 1\}$, where we can estimate the negative effect of $\overline\vp$ as follows
    \begin{equation*}
      \begin{aligned}
        \log\big(& V^{liq}_{\tau}(\overline{\vp}^0,\overline{\vp}^1)\big)-\log\big(V^{liq}_{\tau}(\hvp^0,\hvp^1)\big) \\
        &\geq \log\big(V^{liq}_{\tau}(\overline{\vp}^0,\overline{\vp}^1)\big)
               = \log\big((\tfrac{1}{\lambda}-1)(e^\tau-1)\big)\\
        &\geq \log\big((\tfrac{1}{\lambda}-1)\tau\big) \geq \log\big((\tfrac{1}{\lambda}-1)\sigma\wedge 1\big),
      \end{aligned}
    \end{equation*}
    where $\sigma:=\inf\{t>0~|~W^1_t\leq 0\}\leq \tau$ for $W_0^1 = 1$. 
   As
    \begin{align*}
      0 &\geq \E\left[\log\left((\tfrac{1}{\lambda}-1)\sigma\wedge 1\right)\right] \\
        & =\int_0^{\frac{\lambda}{1-\lambda}}\log\big(\tfrac{1-\lambda}{\lambda}z\big)\left(\tfrac{1}{2\pi z^3}\right)^{\frac{1}{2}}\exp\left(-\tfrac{(z-1)^2}{2z}\right)dz \\
        & =:-c_2>-\infty,
    \end{align*}
    we obtain for the negative effect that
    \begin{equation}\label{pr:ml:est4}
      \E\Big[\Big(\log\big(V^{liq}_{\tau}(\overline{\vp}^0,\overline{\vp}^1)\big)-\log\big(V^{liq}_{\tau}(\hvp^0,\hvp^1)\big)\Big)
                  \mathbf{1}_{\{\vr=\infty\}}\Big]\geq -c_2 \p[\vr=\infty].
    \end{equation}
   Combining \eqref{pr:ml:est3} and \eqref{pr:ml:est4} then gives
     $$ \E\Big[\log\big(V^{liq}_{\tau}(\overline{\vp}^0,\overline{\vp}^1)\big)-\log\big(V^{liq}_{\tau}(\hvp^0,\hvp^1)\big)\Big]
           \geq c_12^{-2n}\p[\vr>2^{-n}] - c_2 \p[\vr=\infty]$$
    and finally
     $$ \E\Big[\log\big(V^{liq}_{\tau}(\overline{\vp}^0,\overline{\vp}^1)\big)-\log\big(V^{liq}_{\tau}(\hvp^0,\hvp^1)\big)\Big]
           \geq\left( c_1  - c_2 \delta\right) 2^{-2n}  \p[\vr>2^{-n}]>0$$
    by \eqref{B4}, as $\delta$ can be chosen arbitrarily small. 
 \end{proof}

\end{appendix}

 \bibliography{ShadowPriceContinuous}

\begin{thebibliography}{10}

\bibitem{AB}
C.~D. Aliprantis and K.~C. Border.
\newblock {\em {Infinite Dimensional Analysis, A Hitchhiker's Guide}}.
\newblock Springer-Verlag, third edition edition, 2006.

\bibitem{BY13}
E.~Bayraktar and X.~Yu.
\newblock {On the Market Viability under Proportional Transaction Costs}.
\newblock {\em Preprint}, 2013.

\bibitem{BCKMK13}
G.~Benedetti, L.~Campi, J.~Kallsen, and J.~Muhle-Karbe.
\newblock {On the existence of shadow prices}.
\newblock {\em Finance Stoch.}, 17(4):801--818, 2013.

\bibitem{BM03}
B.~Bouchard and L.~Mazliak.
\newblock {A multidimensional bipolar theorem in $L^0(\R^d;\Omega,\mathcal
  F,P)$}.
\newblock {\em Stochastic Process. Appl.}, 107(2):213--231, 2003.

\bibitem{BT11}
B.~Bouchard and N.~Touzi.
\newblock {Weak dynamic programming principle for viscosity solutions}.
\newblock {\em SIAM Journal on Control and Optimization}, 107(2):948--962,
  2011.

\bibitem{CS06}
L.~Campi and W.~Schachermayer.
\newblock {A super-replication theorem in {K}abanov's model of transaction
  costs}.
\newblock {\em Finance Stoch.}, 10(4):579--596, 2006.

\bibitem{CSZ13}
J.~Choi, M.~Sirbu, and G.~\v{Z}itkovi{\'c}.
\newblock {Shadow prices and well-posedness in the problem of optimal
  investment and consumption with transaction costs}.
\newblock {\em SIAM Journal of Control and Optimization}, 51(6):4414--4449,
  2013.

\bibitem{CS96}
T.~Choulli and C.~Stricker.
\newblock {Deux applications de la d{\'e}composition de
  Galtchouk--Kunita--Watanabe}.
\newblock {\em S{\'e}minaire de Probabilit{\'e}s XXX, Springer Lecture Notes in
  Mathematics}, 1626:12--23, 1996.

\bibitem{CK96}
J.~Cvitani{\'c} and I.~Karatzas.
\newblock {Hedging and portfolio optimization under transaction costs: a
  martingale approach}.
\newblock {\em Math. Fin.}, 6(2):113--165, 1996.

\bibitem{CMKS14}
C.~Czichowsky, J.~Muhle-Karbe, and W.~Schachermayer.
\newblock {Transaction Costs, Shadow Prices, and Duality in Discrete Time}.
\newblock {\em SIAM Journal on Financial Mathematics}, 5(1):258--277, 2014.

\bibitem{CS14}
C.~Czichowsky and W.~Schachermayer.
\newblock {Duality theory for portfolio optimisation under transaction costs}.
\newblock {\em Preprint}, 2014.

\bibitem{CS13}
C.~Czichowsky and W.~Schachermayer.
\newblock {Strong supermartingales and limits of non-negative martingales}.
\newblock {\em to appear in Annals of Probability}, 2014.

\bibitem{DPT01}
G.~Deelstra, H.~Pham, and N.~Touzi.
\newblock {Dual formulation of the utility maximization problem under
  transaction costs}.
\newblock {\em Ann. Appl. Probab.}, 11(4):1353--1383, 2001.

\bibitem{DS94}
F.~Delbaen and W.~Schachermayer.
\newblock {A general version of the fundamental theorem of asset pricing}.
\newblock {\em Mathematische Annalen}, 300:463--520, 1994.

\bibitem{FK97}
H.~F{\"o}llmer and D.~Kramkov.
\newblock {Optional decompositions under constraints}.
\newblock {\em Prob. Theory Related Fields}, 109(1):1--25, 1997.

\bibitem{GGMKS14}
S.~Gerhold, P.~Guasoni, J.~Muhle-Karbe, and W.~Schachermayer.
\newblock {Transaction Costs, Trading Volume, and the Liquidity Premium}.
\newblock {\em Finance Stoch.}, 18(1):1--37, 2014.

\bibitem{GMKS13}
S.~Gerhold, J.~Muhle-Karbe, and W.~Schachermayer.
\newblock {The dual optimizer for the growth-optimal portfolio under
  transaction costs}.
\newblock {\em Finance Stoch.}, 17(2):325--354, 2013.

\bibitem{GK03}
T.~Goll and J.~Kallsen.
\newblock {A complete explicit solution to the log-optimal portfolio problem}.
\newblock {\em Ann. Appl. Probab.}, 13:774--799, 2003.

\bibitem{G06}
P.~Guasoni.
\newblock {No arbitrage under transaction costs, with fractional Brownian
  motion and beyond}.
\newblock {\em Math. Fin.}, 16(3):569--582, 2006.

\bibitem{GRS10}
P.~Guasoni, M.~Rasonyi, and W.~Schachermayer.
\newblock {The fundamental theorem of asset pricing for continuous processes
  under small transaction costs}.
\newblock {\em Annals of Finance}, 6:157--191, 2010.

\bibitem{HP11}
A.~Herczegh and V.~Prokaj.
\newblock {Shadow price in the power utility case}.
\newblock {\em Preprint}, 2011.

\bibitem{JK95}
E.~Jouini and H.~Kallal.
\newblock {Martingales and arbitrage in securities markets with transaction
  costs}.
\newblock {\em J. Econom. Theory}, 66(1):178--197, 1995.

\bibitem{KMK10}
J.~Kallsen and J.~Muhle-Karbe.
\newblock {On Using Shadow Prices in Portfolio Optimization with Transaction
  Costs}.
\newblock {\em Ann. Appl. Probab.}, 20:1341--1358, 2010.

\bibitem{KMK11}
J.~Kallsen and J.~Muhle-Karbe.
\newblock {Existence of shadow prices in finite probability spaces}.
\newblock {\em Math. Methods Oper. Res.}, 73(2):251--262, 2011.

\bibitem{KK07}
I.~Karatzas and C.~Kardaras.
\newblock {The num{\'e}raire portfolio in semimartingale financial models}.
\newblock {\em Finance Stoch.}, 11(4):447--493, 2007.

\bibitem{K12}
C.~Kardaras.
\newblock {Market viability via absence of arbitrage of the first kind}.
\newblock {\em Finance and Stochastics}, 16(4):651--667, 2012.

\bibitem{KS99}
D.~Kramkov and W.~Schachermayer.
\newblock {The asymptotic elasticity of utility functions and optimal
  investment in incomplete markets}.
\newblock {\em Ann. Appl. Probab.}, 9(3):904--950, 1999.

\bibitem{L00}
M.~Loewenstein.
\newblock {On optimal portfolio trading strategies for an investor facing
  transactions costs in a continuous trading market}.
\newblock {\em Journal of Mathematical Economics}, 33:209--228, 2000.

\bibitem{N06}
A.~Nikeghbali.
\newblock {An essay on the general theory of stochastic processes}.
\newblock {\em Probability Survey}, 300:345--412, 2006.

\bibitem{RW94}
L.~Rogers and D.~Williams.
\newblock {\em {Diffusion, Markov Processes and Martingales}}.
\newblock John Wiley \& Sons, 2nd edition, 1994.

\bibitem{WS01}
W.~Schachermayer.
\newblock {Optimal investment in incomplete markets when wealth may become
  negative}.
\newblock {\em Ann. Appl. Probab.}, 11(3):694--734, 2001.

\bibitem{WS14}
W.~Schachermayer.
\newblock {Admissible trading strategies under transaction costs}.
\newblock {\em S{\'e}minaire de Probabilit{\'e}s XLVI. Springer Lecture Notes
  in Mathematics}, 2123:317--331, 2014.

\bibitem{SY98}
C.~Stricker and J.-A. Yan.
\newblock {Some remarks on the optional decomposition theorem}.
\newblock {\em S{\'e}minaire de Probabilit{\'e}s XXXII. Springer Lecture Notes
  in Mathematics}, 1686:56--66, 1998.

\bibitem{TS14}
K.~Takaoka and M.~Schweizer.
\newblock {A note on the condition of no unbounded profit with bounded risk}.
\newblock {\em Finance Stoch.}, 18:393--405, 2014.

\bibitem{Z10}
G.~\v{Z}itkovi{\'c}.
\newblock {Convex compactness and its applications}.
\newblock {\em Math. Fin. Economics}, 300:1--12, 2010.

\end{thebibliography}
 \bibliographystyle{abbrv}

\end{document}